\documentclass[letterpaper,twocolumn,10pt]{article}
\usepackage{usenix2019_v3}

\usepackage{tikz}
\usepackage{comment}

\usepackage[camera]{dtrt}

\hypersetup{
     colorlinks=true,
     linkcolor=blue,
     filecolor=blue,
     citecolor = blue,
     urlcolor=blue,
 }

\usepackage[T1]{fontenc}
\usepackage{amssymb, amsmath, amsthm, booktabs}
\usepackage{graphicx}
\usepackage[capitalize]{cleveref}
\usepackage{enumitem}
\usepackage{subcaption}
\usepackage{dashbox}
\usepackage{cancel}
\usepackage{float}
\usepackage{csquotes}
\usepackage{listings}
\usepackage{algorithm}
\usepackage[noend]{algpseudocode}
\usepackage[normalem]{ulem}

\usepackage{array}
\newcolumntype{H}{>{\setbox0=\hbox\bgroup}c<{\egroup}@{}}

\newcommand{\myparagraph}[1]{\smallskip\noindent\textbf{#1.}}

\crefname{part}{\S}{\S\S}
\crefname{chapter}{\S}{\S\S}
\crefname{section}{\S}{\S\S}
\crefname{subsection}{\S}{\S\S}

\crefname{claim}{Claim}{Claims}
\crefname{remark}{Remark}{Remarks}

\newcommand{\pname}{\texorpdfstring{\ensuremath{\mathsf{CRATE}}}{CRATE}}

\begin{document}

\definecolor{chain}{rgb}{0.0, 0.0, 1.0} %
\definecolor{dag}{rgb}{0.54, 0.2, 0.14} %
\definecolor{comment}{gray}{0.5} %

\algblock{Variables}{EndVariables}
\algnewcommand\And{\textbf{and}}
\algnewcommand\Or{\textbf{or}}
\algnewcommand\Not{\textbf{not}}

\newcommand{\hash}{H}
\newcommand{\gen}{\mathcal{G}}
\newcommand{\prover}{\mathcal{P}}
\newcommand{\verifier}{\mathcal{V}}

\newcommand{\vsm}[1]{\mathsf{VSM}_{#1}}
\newcommand{\bc}[1]{bc_{#1}}
\newcommand{\rol}[1]{\mathsf{rol}_{#1}}

\newcommand{\protocol}{\Pi}

\newcommand{\pstat}{\mathsf{status}}

\newcommand{\state}[1]{\mathsf{ST_{#1}}}
\newcommand{\dst}[1]{h_{\mathsf{ST_{#1}}}}
\newcommand{\dstt}[1]{h_{\mathsf{ST_{#1}'}}}
\newcommand{\dstemp}[1]{h_{\mathsf{ST_{#1}}}^{\mathsf{temp}}}

\newcommand{\dsto}[1]{h_{\mathsf{ST_{#1}}}^{o}}
\newcommand{\dstto}[1]{h_{\mathsf{ST_{#1}'}}^{o}}

\newcommand{\nonces}{\mathsf{nonces}}
\newcommand{\tx}{\mathsf{tx}}
\newcommand{\roots}{\mathsf{roots}}

\newcommand{\addr}{\mathsf{addr}}
\newcommand{\cdata}{\mathsf{cdata}}

\newcommand{\contract}[1]{\mathcal{SC}_{#1}} %

\newcommand{\call}{\mathsf{call}}

\newcommand{\msg}{\mathsf{msg}}
\newcommand{\sender}{\mathsf{sender}}

\newcommand{\addrr}[1]{\mathsf{addr}_{#1}}
\newcommand{\cdataa}[1]{\mathsf{cdata}_{#1}}

\newcommand{\status}[1]{\mathsf{status_{#1}}}
\newcommand{\cp}[1]{\mathsf{cp_{#1}}}
\newcommand{\flag}[1]{\mathsf{flag_{#1}}}

\newcommand{\trigcalls}{\mathsf{trigCalls}}
\newcommand{\trigcall}{\mathsf{trigCall}}
\newcommand{\id}{\mathsf{id}}
\newcommand{\rolids}{\mathsf{rolIDs}}

\newcommand{\valproof}{\mathsf{vProof}}
\newcommand{\valpr}[1]{\mathsf{vProof}_{#1}}

\newcommand{\locevd}{\mathsf{locEvd}}
\newcommand{\pcevd}{\mathsf{pcEvd}}
\newcommand{\cevd}{\mathsf{cEvd}}
\newcommand{\abevd}{\mathsf{abEvd}}

\newcommand{\gsc}{\mathsf{GSC}}
\newcommand{\trig}{\mathsf{trigger}}
\newcommand{\act}{\mathsf{action}}

\newcommand{\ttree}{\mathsf{trigTree}}
\newcommand{\atree}{\mathsf{actTree}}
\newcommand{\troot}[1]{\mathsf{trigRt}_{\mathsf{#1}}}
\newcommand{\aroot}[1]{\mathsf{actRt}_{\mathsf{#1}}}

\newcommand{\bridgeproof}{\mathsf{bridgeProof}}
\newcommand{\attproofs}{\mathsf{attProofs}}

\newcommand{\leader}{\mathsf{leader}}
\newcommand{\vsmid}{\mathsf{id}}

\newcommand{\require}{\textbf{Require }}

\newcommand{\batch}[1]{\mathsf{batch}_{#1}}
\newcommand{\batches}[1]{\mathsf{batches}_{#1}}

\newcommand{\actexec}[1]{\mathsf{actionExec}_{#1}}
\newcommand{\atomexec}{\mathsf{AtomExec}}

\newcommand{\code}[1]{\mathsf{trace}_{#1}}
\newcommand{\desc}[1]{\mathsf{desc}_{#1}}
\newcommand{\descnext}[1]{\mathsf{desc_{next}}_{#1}}
\newcommand{\nul}{\mathsf{null}}

\newcommand{\idx}{\mathsf{index}}
\newcommand{\paired}{\mathsf{Paired}}
\newcommand{\free}{\mathsf{Free}}

\newcommand{\commit}{\mathsf{Commit}}
\newcommand{\abort}{\mathsf{Abort}}
\newcommand{\decisions}{\mathsf{decisions}}

\newcommand{\true}{\mathsf{True}}
\newcommand{\false}{\mathsf{False}}
\newcommand{\nott}{\mathsf{NOT}}

\newcommand{\insertt}{\mathsf{insert}}

\newcommand{\tnonce}[1]{\mathsf{tNonce_{#1}}}
\newcommand{\anonce}[1]{\mathsf{aNonce_{#1}}}
\newcommand{\xsender}[1]{\mathsf{xSender_{#1}}}
\newcommand{\msghash}{\mathsf{msgHash}}
\newcommand{\msgsender}{\mathsf{msg.sender}}

\newcommand{\enonce}[1]{\mathsf{entryNonce_{#1}}}
\newcommand{\snonce}[1]{\mathsf{sessionNonce_{#1}}}
\newcommand{\sactive}[1]{\mathsf{sessionActive_{#1}}}
\newcommand{\sid}[1]{\mathsf{sessionID_{#1}}}

\newcommand{\tcalled}[1]{\mathsf{triggerCalled_{#1}}}
\newcommand{\acalled}[1]{\mathsf{actionCalled_{#1}}}

\newcommand{\tdata}[1]{\mathsf{trigData)_{#1}}}
\newcommand{\tdatalist}[1]{\mathsf{trigDataList)_{#1}}}

\newcommand{\crt}{\mathsf{crt}}
\newcommand{\ccrt}{\mathsf{crt_{chain}}}
\newcommand{\dagcrt}{\mathsf{crt}_{DAG}}
\newcommand{\burn}{\mathsf{burn}}
\newcommand{\mint}{\mathsf{mint}}

\newcommand{\outt}[1]{\mathsf{out}_{#1}}

\newcommand{\ovsm}{\mathsf{oVSM}}
\newcommand{\rolset}{\mathcal{R}}

\newtheorem{theorem}{Theorem}
\newtheorem{definition}{Definition}
\newtheorem{remark}{Remark}
\newtheorem{lemma}{Lemma}
\newtheorem{corollary}{Corollary}
\newtheorem{proposition}{Proposition}
\newtheorem{claim}{Claim}
\newtheorem{observation}{Observation}
\newtheorem{fact}{Fact}
\newtheorem{assumption}{Assumption}

\definecolor{verylightgray}{rgb}{.97,.97,.97}

\lstdefinelanguage{Solidity}{
	keywords=[1]{anonymous, assembly, assert, balance, break, call, callcode, case, catch, class, constant, continue, constructor, contract, debugger, default, delegatecall, delete, do, else, emit, event, experimental, export, external, false, finally, for, function, gas, if, implements, import, in, indexed, instanceof, interface, internal, is, length, library, log0, log1, log2, log3, log4, memory, modifier, new, payable, pragma, private, protected, public, pure, push, require, return, returns, revert, selfdestruct, send, solidity, storage, struct, suicide, super, switch, then, this, throw, transfer, true, try, typeof, using, value, view, while, with, addmod, ecrecover, keccak256, mulmod, ripemd160, sha256, sha3}, %
	keywordstyle=[1]\color{blue}\bfseries,
	keywords=[2]{address, bool, byte, bytes, bytes1, bytes2, bytes3, bytes4, bytes5, bytes6, bytes7, bytes8, bytes9, bytes10, bytes11, bytes12, bytes13, bytes14, bytes15, bytes16, bytes17, bytes18, bytes19, bytes20, bytes21, bytes22, bytes23, bytes24, bytes25, bytes26, bytes27, bytes28, bytes29, bytes30, bytes31, bytes32, enum, int, int8, int16, int24, int32, int40, int48, int56, int64, int72, int80, int88, int96, int104, int112, int120, int128, int136, int144, int152, int160, int168, int176, int184, int192, int200, int208, int216, int224, int232, int240, int248, int256, mapping, string, uint, uint8, uint16, uint24, uint32, uint40, uint48, uint56, uint64, uint72, uint80, uint88, uint96, uint104, uint112, uint120, uint128, uint136, uint144, uint152, uint160, uint168, uint176, uint184, uint192, uint200, uint208, uint216, uint224, uint232, uint240, uint248, uint256, var, void, ether, finney, szabo, wei, days, hours, minutes, seconds, weeks, years},	%
	keywordstyle=[2]\color{teal}\bfseries,
	keywords=[3]{block, blockhash, coinbase, difficulty, gaslimit, number, timestamp, msg, data, gas, sender, sig, value, now, tx, gasprice, origin},	%
	keywordstyle=[3]\color{violet}\bfseries,
	identifierstyle=\color{black},
	sensitive=true,
	comment=[l]{//},
	morecomment=[s]{/*}{*/},
	commentstyle=\color{gray}\ttfamily,
	stringstyle=\color{red}\ttfamily,
	morestring=[b]',
	morestring=[b]"
}

\lstset{
	language=Solidity,
	backgroundcolor=\color{verylightgray},
	extendedchars=true,
	basicstyle=\footnotesize\ttfamily,
	showstringspaces=false,
	showspaces=false,
	numbers=none,
	numberstyle=\footnotesize,
	numbersep=9pt,
	tabsize=2,
	breaklines=true,
	showtabs=false,
	captionpos=b
}

\date{}

\title{\Large \bf \pname: Cross-Rollup Atomic Transaction Execution}

\author{
{\rm Ioannis Kaklamanis}\\
Yale University
\and
{\rm Fan Zhang}\\
Yale University
} %

\maketitle

\begin{abstract}
Blockchains have revolutionized decentralized applications, with composability enabling atomic, trustless interactions across smart contracts. However, layer 2 (L2) scalability solutions like rollups introduce fragmentation and hinder composability.
Current cross-chain protocols, including atomic swaps, bridges, and shared sequencers, lack the necessary coordination mechanisms or rely on trust assumptions, and are thus not sufficient to support full cross-rollup composability.
This paper presents $\pname$, a secure protocol for cross-rollup composability that ensures all-or-nothing and serializable execution of cross-rollup transactions (CRTs).
 $\pname$ supports rollups on distinct layer 1 (L1) chains, achieves finality in 4 rounds on L1, and only relies on the underlying L1s and the liveness of L2s. We introduce two formal models for CRTs, define atomicity within them, and formally prove the security of $\pname$. We also provide an implementation of $\pname$ along with a cross-rollup flash loan application; our experiments demonstrate that $\pname$ is practical in terms of gas usage on L1.
\end{abstract}

\section{Introduction}~\label{sec:intro}
Blockchains have revolutionized trust, democratization, and privacy for applications in finance and beyond. As of October 2024, blockchain systems and applications collectively control assets worth over 2.5 trillion dollars~\cite{forbes-crypto-cap-24}. A unique attribute that sets blockchain applications apart from conventional ones is {\em composability}, the ability to combine services from multiple smart contracts in a single atomic execution. Composability enables trustless, atomic interactions that lead to innovative applications, such as flash loans~\cite{qin-attacking-defi-flash-loans-21}. 

Blockchains like Ethereum~\cite{wood-ethereum-2014} have a ``layer 2''-centric~\cite{eth-org-layer-2} roadmap for scalability. Layer 2 (L2) chains are separate blockchains built on top of the layer 1 (L1) chain, inheriting some of its security features. A rollup, a specific type of L2 chain, handles transaction execution off-chain and then submits the transaction data back to the main Ethereum blockchain. Even though rollups are envisioned to host the execution of more and more transactions, the scalability provided by L2 chains is {\em at odds} with composability. 
L2 systems like rollups are isolated from each other and forgo the advantages~\cite{a16z-composability} of composability in a cross-rollup context.

The fragmentation barrier poses a distinct challenge: the necessity for various chains to interoperate (IO), i.e., securely communicate and coordinate. Several cross-chain interoperability (IO) protocols have surfaced to address this gap, primarily adopting the framework of atomic swaps and cross-chain bridges (see, e.g.,~\cite{augusto-sok-io-24} for a survey). These systems have found extensive use in practice. 
For instance, bridges~\cite{xiew-zkbridge-22, axelar, near-bridge} move more than $1$ billion USD worth of assets cross-chain per week within Decentralized Finance (DeFi) services alone, as of the time of writing~\cite{defillama-bridges}. 
Atomic swaps~\cite{bitcoin-wiki-htlc,herlihy-atomic-cc-swaps-18} are also widely deployed, with which two parties can exchange their assets on separate blockchains with an atomicity guarantee: either they both receive the other party’s assets, or neither does. 
However, while existing IO protocols enable services such as asset transfers, they cannot offer full cross-rollup {\em composability}, since they lack a coordination mechanism. 
Notably, they cannot support cross-rollup flash loans or other complex applications that involve executing multiple transactions across rollups atomically.

A few academic proposals~\cite{lu-atomic-cc-interactions-24, fal-tccsci-23, zakhary-ac3-20} achieve composability for transactions spanning multiple L1 chains. However, they tackle a fundamentally different problem, as they operate on L1; they further require complex state locking mechanisms and multiple rounds, thus prohibiting certain applications like cross-chain flash loans. 
Shared sequencer networks, such as Espresso, Astria, and Radius~\cite{espresso-docs, astria-docs, radius-docs}, enable certain types of cross-rollup transactions but crucially trust the sequencer for proper sequencing to guarantee atomicity.

\myparagraph{The problem}
We investigate protocols for executing cross-rollup transactions (CRTs) atomically, enabling users to execute a sequence of transactions across multiple rollups in an all-or-nothing manner. 
Besides the all-or-nothing feature, we also require \emph{serializability} (in particular, \emph{weak serializability}), i.e., that the execution preserves the order of transactions within the CRT. In contrast, \emph{full serializability} additionally requires that no other transaction can interleave between CRT transactions (see \cref{sec:discussion} for details). \done Serializability is needed by most applications, such as arbitrage consisting of carefully ordered buy-and-sell transactions. 
Candidate solutions should also minimize trust in third parties and be flexible by supporting rollups that reside on different L1s.

\myparagraph{Building block: Shared Validity Sequencing}
Shared Validity Sequencing (SVS)~\cite{shared-val-seq-23} proposes a trigger-action paradigm for a smart contract on one rollup to remotely invoke a smart contract method on another rollup. A \textit{smart contract}~\cite{eth-org-smart-contracts} is a program stored on a blockchain that gets executed once all its predetermined conditions are satisfied. A rollup is operated by an \emph{Executor} and is controlled by a Validator smart contract ($\vsm{}$) on L1. SVS relies on a shared Executor between two rollups and a special system contract ($\gsc$) residing on both rollups. $\gsc$ has two main functions,
$\trig$ (called by users) and $\act$ (called by the Executor), and two Merkle trees, $\ttree$ and $\atree$. 
The atomicity of CRTs is ensured by comparing the roots of the two Merkle trees, enforced by the Validator smart contract, avoiding trust in any third party. A more detailed description of SVS can be found in \cref{sec:overview-svs}.

The SVS solution has two major problems. If the user CRT has a length greater than $2$, a malicious Executor could reorder transactions such that the all-or-nothing check passes, yet serializability breaks. Restoring serializability while preserving expressiveness is a nontrivial challenge we tackle in this paper.
Moreover, SVS also assumes that rollups share the same Executor and reside on the same L1 chain, which lacks flexibility. A more flexible protocol would allow rollups to reside on separate chains and be operated by distinct Executors.

\myparagraph{Our approach}
We present our protocol, $\pname$: \textbf{C}ross-\textbf{R}ollup \textbf{A}tomic \textbf{T}ransaction \textbf{E}xecution, which extends the existing rollup architecture and addresses the two key problems found in~\cite{shared-val-seq-23}, thus enabling cross-rollup composability. 
$\pname$ begins with the user creating a single source transaction that encodes her CRT using the trigger-action framework and sends it to the Executor. In Phase 1, the Executor executes a transaction batch off-chain, and generates new state digests and validity proofs. $\pname$ supports serializability and more complex CRTs by modifying the GSC contract to handle multiple triggered actions and track CRT sessions with session nonces. In Phase 2, the Executor drives the two VSMs to complete a Two-Phase-Commit (2PC) protocol, in order to jointly agree on the new pair of state digests. First, one of the two VSMs is designated to be the \emph{leader} and pre-commits to the new state digest, followed by the non-leader VSM. Then the leader VSM commits to the digest only when presented with evidence that the other VSM has pre-committed; finally the non-leader VSM commits, too, after seeing evidence of the leader VSM's commit decision.

\parhead{Formal definition and proof}
The lack of formal definition and rigorous analysis in~\cite{shared-val-seq-23} leaves the flaws unnoticed.
Thus, we formally define cross-rollup transactions (CRTs) in two different programming models, and then define atomicity for CRT execution, providing -- to our knowledge -- the first formal treatment of cross-rollup composability. We use our formal framework to rigorously prove security for $\pname$.

\parhead{Implementation and evaluation.}
We implement $\pname$ on Ethereum networks using the Foundry~\cite{foundry-book} toolkit, focusing on the Executor, L2-based $\gsc$ contract, and L1-based Validator ($\vsm{}$) contract. For evaluation, our implementation using SNARK-based membership proofs incurs a $0.32-0.75\times$ increase in gas usage compared to Zksync Era~\cite{zksync-era}. We additionally implement and evaluate a modular cross-rollup flash loan application on top of $\pname$.
Regarding rollup finality, $\pname$ achieves finality in 4 rounds on L1 (i.e., $~52$ minutes on Ethereum). Existing zk-rollups submit new state roots to L1 at intervals in the order of \emph{hours}, since SNARK proof generation is the bottleneck. Thus, $\pname$'s latency does not affect existing rollups' finality and throughput, since the $\vsm{}$s were already not utilized during that time for any new state updates.

\subsection*{Our Contributions.}
\begin{itemize}[leftmargin=*]
    \item We analyze the state-of-the-art cross-rollup composability paradigm~\cite{shared-val-seq-23} and demonstrate a concrete attack that breaks its security for CRTs of length greater than $2$. 
    \item We provide two formal models for cross-rollup transactions (CRTs): the chain-CRT and DAG-CRT programming model. Within these models, we formally define atomicity for CRTs, providing the first such formal security treatment.
    \item We introduce $\pname$, a secure cross-rollup composability protocol that (1) only trusts the underlying L1s, the liveness of the L2s, and the liveness of a bridge between the L1s, and (2) achieves finality in 4 rounds on L1. We then formally prove the security of $\pname$.
    \item We implement $\pname$ along with an end-to-end cross-rollup flash loan. Our experiment shows that $\pname$ is practical in terms of L1 gas usage.
    \item We introduce extensions to $\pname$ which enable full serializability and support distinct Executors.
\end{itemize}

\section{Background}~\label{sec:background}

\myparagraph{Blockchains}
A blockchain is a distributed protocol where a group of nodes collectively maintains an ordered list of blocks. A block is a data structure that stores a header and a list of transactions, as well as important metadata, such as a pointer to the previous block. 

\myparagraph{Smart contracts}
Many blockchains support expressive programs called smart contracts; these are stateful programs whose state persists on the blockchain. Our solution relies heavily on smart contracts on both L1 and L2, so we introduce some commonly used notation. The keyword \textbf{require} is used to check whether a certain predicate is true; if not, then the EVM reverts execution of the transaction. A smart contract can emit \emph{events} using the \textbf{emit} keyword; events provide non-persistent logging functionality which can be used by upper-level applications.

\myparagraph{Merkle Tree} A Merkle tree~\cite{merkle-tree} is a dynamic data structure widely used to commit to an ordered sequence of values. Under the assumption of cryptographic hash functions, two Merkle trees are equal if and only if their roots are equal. Ethereum uses the more efficient variant of Merkle-Patricia Tries~\cite{mpt-trie} for storage. In this work we use Merkle trees to mean both regular Merkle trees and MPT tries.

\newcommand{\scstruct}[1]{\mathsf{SC}_{#1}}

\myparagraph{Rollups: Executor and VSM} We consider validity-proof-based rollups, widely known as zk-rollups~\cite{eth-org-zk-rollups}. Each rollup is operated by a single entity -- which we call Executor~\footnote{In practice, there are distinct entities (Sequencer, Executor, Prover, etc.) assuming different roles in the rollup operation.} -- and is controlled by a Validator smart contract ($\vsm{}$) on layer $1$. The Executor is responsible for receiving, ordering, and executing transactions submitted by rollup users. The Executor also creates a validity proof for the state transition, and submits both the new state digest and the proof to the $\vsm{}$ (at the same time). The $\vsm{}$ contract maintains a short digest $\dst{}$ of the current state $\mathsf{ST}$ of the rollup. The $\vsm{}$ contract exposes a function $\textproc{updateDigest}(\cdot)$, which advances the current state digest and can only be called by the Executor. This function verifies the validity proof before accepting the new state digest. Suppose $\dst{}$ is the current state digest stored by a rollup VSM, and let $\dstt{} \neq \dst{}$ be another state digest. We say that $\dstt{}$ is the \emph{immediate next accepted} digest after $\dst{}$ if $\dstt{}$ is the first digest submitted to and accepted by the VSM contract through the $\textproc{updateDigest}(\cdot)$ function. 

\myparagraph{Bridges}~\label{par:background-bridges}
$\pname$ relies on cross-chain bridges for communication between VSMs and can use any trustless bridge, such as zkBridge~\cite{xiew-zkbridge-22}. We abstract bridges as follows: a contract $\contract{1}$ on $\bc{1}$ can access state from $\contract{2}$ on $\bc{2}$ via a tuple $(\scstruct{2}, \attproofs)$, where $\scstruct{2}$ represents $\contract{2}$'s state, and $\attproofs$ provides membership proofs against a relayed $\bc{2}$ block header.

\myparagraph{Application: Flash Loan}
Flash loans~\cite{qin-attacking-defi-flash-loans-21} enable a user to borrow a loan and repay it in the same blockchain transaction; if the loan is not repaid in the same transaction, then the entire transaction will fail. They are a type of uncollateralized loans, available to all users regardless of the capital they own.

\myparagraph{Two-Phase Commit (2PC)}
Two-Phase Commit (2PC) is a protocol introduced in the setting of distributed databases~\cite{bernstein-concurrency-databases-1986} to ensure atomicity in transactions across multiple servers. It operates in two phases: the pre-commit phase, where the coordinator asks all participants if they can commit, and the commit phase, where the coordinator either commits or aborts the transaction based on responses. This protocol guarantees that either all participants commit the transaction or none do.

\myparagraph{Zero-Knowledge Arguments and SNARKs}
An argument system for an NP relation $R$ is a protocol between a computationally-bounded prover $\prover$ and a verifier $\verifier$. Through this protocol, $\verifier$ is convinced by $\prover$ that a witness $w$ exists such that $(x,w) \in R$ for some input $x$. There is also an algorithm $\gen$ which produces public parameters during the setup phase. A triple $(\gen, \prover, \verifier)$ is a zero-knowledge argument (ZKA) of knowledge for $R$ is it satisfies \emph{completeness}, \emph{knowledge soundness}, and \emph{zero knowledge}; we refer the reader to~\cite{groth-16, thaler-proofs-args-zk-22} for formal definitions of these properties. $(\gen, \prover, \verifier)$ is considered a \emph{succint} argument system if the total communication (proof size) between $\prover$ and $\verifier$ as well as the running time of $\verifier$ are bounded by $\mathsf{poly}(\lambda, |x|, \log|R|)$, where $\log|R|$ represents the size of the circuit computing $R$. For this paper, we require only a succinct non-interactive argument of knowledge (SNARK) that satisfies the first two properties and ensures succinctness.

\section{Technical Overview}~\label{sec:overview}

In this section, we describe our problem statement, review SVS~\cite{shared-val-seq-23}, and present a high-level overview of $\pname$.

\subsection{Problem Statement}~\label{sec:problem-statement}
We investigate protocols for the problem of executing cross-rollup transactions (CRTs) atomically. 
In this paper, we use ``rollups'' to refer to ``zk-rollups,'' systems that use SNARKs to prove the correctness of transaction execution.

\myparagraph{Assumptions, Threat Model, and Goals}
 In this work, we consider the case of $2$ rollups, which reside on distinct L1 chains. For ease of exposition, we assume the participating rollups are operated by the same (shared) Executor. Later (see~\cref{sec:discussion}), we explain how our protocol can serve for the case of rollups with \textbf{different Executors}. 
  We adopt the standard assumption that participating L1 chains are live and safe (informally, liveness means the blockchain keeps processing transactions, and safety means that the finalized blocks are consistent cross replicas; see, e.g.,~\cite{shiFoundations2020} for formal definitions.)
 We consider a computationally bounded adversary which can corrupt the Executor in order to break atomicity of users' cross-rollup transactions; he cannot, however, tamper with user-signed transactions. We assume the Executor maintains the rollup's liveness, i.e., he is always online, processing transactions on L2 and submitting state roots on L1.~\footnote{This is a standard assumption for existing rollups, which employ stake-based incentive mechanisms and replace misbehaving Executors.} 

 The goal is to allow a user to execute a sequence of transactions spanning multiple rollups in an atomic (i.e., all-or-nothing) manner while preserving the ordering of the transactions within each rollup.
 We aim to design a \textbf{C}ross-\textbf{R}ollup \textbf{A}tomic \textbf{T}ransactions (CRAT) protocol that \textbf{(a)} is as \textbf{secure} as existing zk-rollups, \textbf{(b)} is \textbf{efficient} in terms of layer 1 transactions, \textbf{(c)} places \textbf{minimal or zero trust} on any third-party, and \textbf{(d)} is \textbf{practical} in terms of on-chain (L1) gas usage.

\myparagraph{Defining Cross-Rollup Transaction (CRT)}
For the rest of this paper, we use the term ``action'' to mean a transaction executed on a single chain. Looking ahead, without formal definition and proof, promising ideas proposed in~\cite{shared-val-seq-23} are in fact insecure. 
Pinning down a formal model and proving security for $\pname$ is a non-trivial technical challenge. 
For the purposes of this overview only, we use the following simpler notations. An \emph{action} can be represented using the format specified by Ethereum Virtual Machine (EVM)~\cite{eth-evm}, i.e., 
$a = (\addr, \cdata)$, where $\addr$ is the smart contract address and $\cdata$ contains the function identifier and arguments.
Then a cross-rollup transaction (CRT) can be viewed as an ordered sequence of actions $ \crt = [a_i]_{i \in [n]}$.

\subsection{Building Block: Trigger-Action Paradigm}~\label{sec:overview-svs}
The state-of-the-art proposal is Shared Validity Sequencing (SVS)~\cite{shared-val-seq-23} and variants (e.g., \cite{espresso-circ-24}).
SVS aims to enable an action $a_1$ on $\rol{1}$ to ``trigger'' a call to action $a_2$ on another rollup $\rol{2}$ (and vice versa and potentially recursively) in an atomic fashion.
SVS assumes that both rollups reside on the same layer 1 and share the same Validator smart contract (which limits its practicality).
To achieve atomicity, the consistency between $a_1$'s input to $a_2$ (which is stored on $\rol{1}$) and $a_2$'s execution (which takes place on $\rol{2}$) must be verified by layer 1. SVS introduced a special layer 2 smart contract to facilitate efficient consistency checks using Merkle Trees.
Specifically, on all rollups $\rol{i}$, a ``general system contract'' $\gsc_i$ is deployed and trusted by all users. $\gsc_i$ exposes two main functions: $\trig(\cdot)$ and $\act(\cdot)$~\footnote{The $\gsc$ function $\act(\cdot)$ is not to be confused with the term ``action'' which represents a blockchain transaction.}. $\gsc$ stores two merkle trees, $\ttree_i$ and $\atree_i$, whose roots we denote by $\troot{i}$ and $\aroot{i}$, respectively.

\begin{figure}
    \centering
    \includegraphics[scale=0.5]{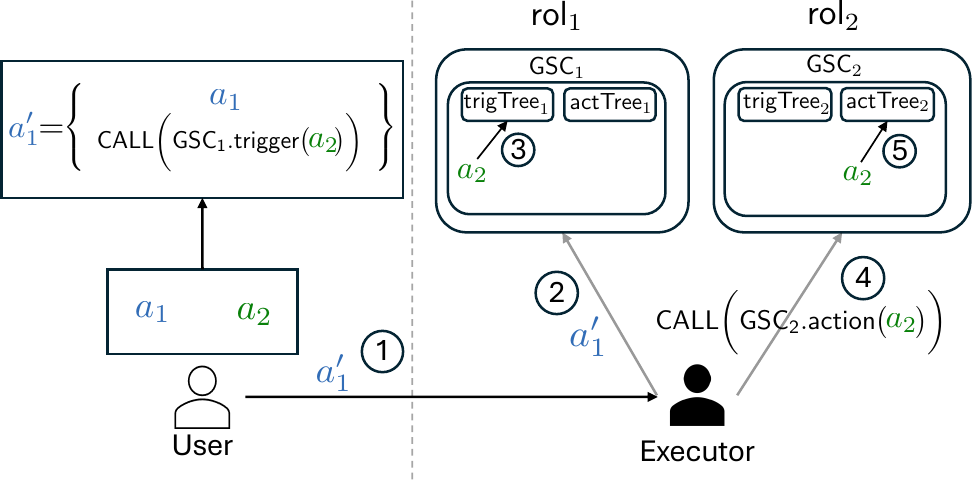}
    \caption{SVS~\cite{shared-val-seq-23} trigger-action paradigm. Here ``$\mathsf{CALL}(f)$'' denotes that the (smart contract) function $f$ will be executed.}
    \label{fig:orig-svs-trigger-action}
\end{figure}

\label{svs-crt-lifecycle} 
In~\cref{fig:orig-svs-trigger-action}, we illustrate the SVS workflow assuming a user who wishes to execute action $a_1 = (\addrr{1}, \cdataa{1})$ on $\rol{1}$ and $a_2 = (\addrr{2}, \cdataa{2})$ on $\rol{2}$ atomically \footnote{For instance, the pair $(a_1,a_2)$ could be a burn-mint token transfer.}.
She first ``bundles'' $(a_1,a_2)$ by creating an \emph{augmented} action $a_1'$ that executes $a_1$ and calls $\gsc_1.\trig(\addrr{2}, \cdataa{2})$ (step 1); then she sends $a_1'$ to the Executor, who in turn executes $a_1'$ on $\rol{1}$ (step 2), causing the tuple $(\addrr{2}, \cdataa{2})$ to be inserted into $\ttree_1$ (step 3). The Executor executes $\gsc_2.\act(\addrr{2}, \cdataa{2})$ on $\rol{2}$ (step 4), and as a result the same tuple $(\addrr{2}, \cdataa{2})$ is inserted into $\atree_2$ (step 5) and $a_2$ is executed.
Now atomicity boils down to a simple consistency check of whether $\troot{1} \overset{?}{=} \aroot{2}$ and $\troot{2} \overset{?}{=} \aroot{1}$, which the L1 $\vsm{}$ can efficiently check.

To see why atomicity holds, note that the execution of $a_1'$ implies (1) $a_1$ is executed, and (2) $(\addrr{2}, \cdataa{2})$ is inserted into $\ttree_1$. Either both (1) and (2) take effect or none does. Similarly, the execution of $\gsc_2.\act(\addrr{2}, \cdataa{2})$ implies (1') $a_2$ is executed on $\rol{2}$, and (2') the same triple $(\addrr{2}, \cdataa{2})$ is inserted into $\atree_2$. $\gsc_2$ ensures that both take place or none does.
Due to the collision resistance of Merkle Trees, the atomicity check above ensures that if (2) and (2') took place, so did (1) and (1').

\subsection{Technical Challenges and Our Solutions}~\label{sec:svs-challenges-and-solutions}

\begin{figure*}
    \centering
    \includegraphics[scale=0.5]{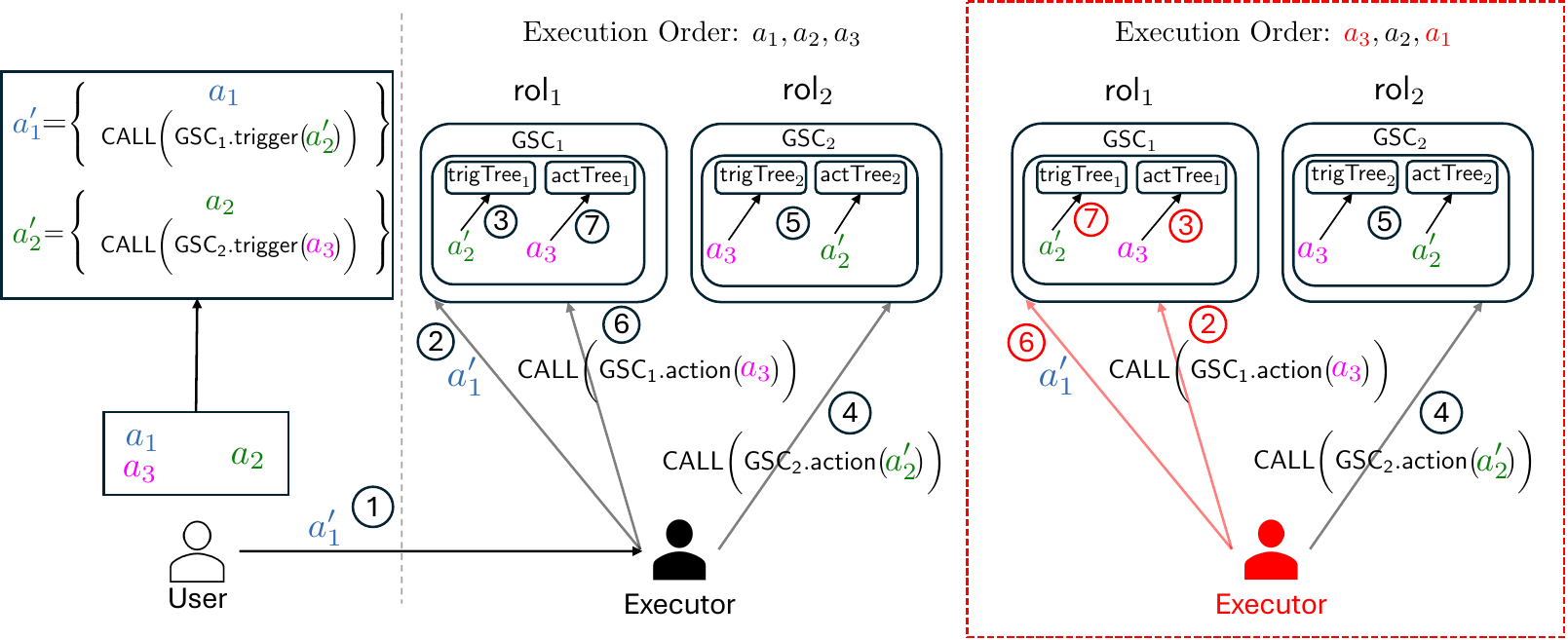}
    \caption{Serializability attack on SVS~\cite{shared-val-seq-23}. A malicious Executor (shown on the right in red) can execute $a_3$ before $a_1'$, flipping the order of steps (2,3) with steps (6,7). Since $a_1'$ only modifies $\ttree_1$ and $\gsc_1.\act(a_3)$ only modifies $\atree_1$, the two $\gsc$ contracts result in the same state as in the honest execution (shown in the middle), thus passing the consistency check.}
    \label{fig:svs-serializability-attack}
\end{figure*}

Although simple and elegant, SVS~\cite{shared-val-seq-23} has two major problems and fixing them requires careful modeling and protocol design. 
First, serializability -- and by extension, atomicity -- breaks down for CRTs of \emph{length greater than $2$}, as illustrated in \cref{fig:svs-serializability-attack}.
Consider $\crt := [a_1, a_2, a_3]$ of length $3$, where $a_1$ and $a_3$ are to be executed on rollup $1$ and $a_2$ on rollup $2$. A natural attempt to do this with SVS is to create augmented actions $(a_1', a_2')$ as above: $a_2'$ is the same as $a_2$ but also triggers $a_3$, and now $a_1'$ is the same as $a_1$ but also triggers $a_2'$ (step 1 in \cref{fig:svs-serializability-attack}).
The intention is that the atomic execution starts with $a_1'$ (steps 2,3), then $a_2'$ (steps 4,5) and $a_3$ (steps 6,7). 
However, a \textit{malicious} Executor can easily mount a re-ordering attack to break atomicity; in our example, he can execute $a_3$ before $a_1'$ on rollup $1$. 
Doing so still passes the consistency check as above and the validity proof verifications because each action was honestly executed. In practice, this can be a serious attack; for example, if $a_1$ and $a_3$ contain arbitrage logic, changing their execution order can cause the user to miss the arbitrage opportunity or end up worse off.

Second, SVS requires all rollups to share the same Executor and reside on the same chain, while many rollups do not satisfy these requirements. Moreover, the Executor is required to submit each new pair of state roots \emph{together} in the same L1 transaction, along with the extended validity proofs described above. Instead, CRAT protocols should ideally be \textit{flexible} and allow rollups to reside on \textit{distinct L1 chains}, as well as be operated by distinct, possibly distrusting Executors.
To overcome this challenge, the Two-Phase Commit (2PC) paradigm is a natural choice and indeed the one we adopt; however, due to the distinct and heterogeneous nature of the entities involved -- two VSM contracts and an untrusted Executor -- applying the 2PC framework tuns out to be highly nontrivial. 

\subsubsection{Our Solutions}
We present $\pname$, along with the careful programming model and formalism needed to tackle the major flaws not addressed by the SVS~\cite{shared-val-seq-23} architecture. 

\myparagraph{Restoring serializability}
We present our new design in two steps. 
In the first step, we fix serializability in the simpler \emph{chain CRT} programming model, in which each action is limited to trigger at most one other action. However, the chain CRT model lacks the necessary expressiveness for certain applications, such as flash loans. We concretely motivate the necessity of such higher expressiveness in~\cref{subsec:application-xfl}. In the second step, we propose a more flexible programming model and modify our serializability solution to support this model.

In more detail, 
we introduce the notion of {\em sessions}. We carefully modify $\gsc$ to track CRT instances by assigning a unique session nonce and maintaining an entry nonce and an active flag.  The source action $(a_1)$ starts a new CRT session by calling a new entry point function in $\gsc$, which increments the session and entry nonces and sets the flag. Subsequent $\act(\cdot)$ calls verify that the session nonce matches the current or next expected value, ensuring consistent use of the same nonce across actions. The VSM contract checks that the sum of the two entry nonces equals the session nonce.

\myparagraph{Supporting DAG CRTs}
We introduce the more expressive Directed Acyclic Graph (DAG) CRT programming model. In this model, each action can trigger multiple subsequent actions, which can in turn trigger additional actions, forming a directed acyclic graph (hence the name ``DAG'').  To integrate this intended functionality within our serializability solution, we represent a DAG CRT as a sequence of actions, similar to chain CRTs, but with the key distinction that each action includes multiple sub-actions to be executed, along with the ability to trigger additional sub-actions.
To support the more expressive DAG CRTs, we modify  $\gsc.\act(\cdot)$ to accept an \emph{array} of triggered actions, instead of a single action. Further, we ensure that all trigger calls and triggered actions within the same transaction share the same nonce value, we repurpose the entrypoint function to be the global entrypoint of $a_1$, all while preserving the previous session-tracking logic.

\myparagraph{Supporting distinct L1 chains}
While the above changes restore serializability, we draw on 2-Phase Commit (2PC) protocols from the distributed computing literature~\cite{bernstein-concurrency-databases-1986} to support rollups with $\vsm{}$s on distinct L1 chains.
The key is to ensure that two VSM contracts accept new state roots in an all-or-nothing fashion. Suppose the Executor is ready to submit a pair of state digests -- one to each VSM -- which were produced by at least one CRT. \done%
In the pre-commit round, one VSM is selected as the leader; it verifies and pre-commits to the submitted state digest and transitions to a paired status; the non-leader VSM then follows the same process and pre-commits only after the leader VSM. In the commit round, the two VSMs commit to the new state digest in the same order after verifying relevant evidence from the other side.

For ease of exposition, we present $\pname$ in the chain-CRT model in \cref{sec:xevm,sec:2pc-trustless}, and present modifications to support the DAG-CRT model in~\cref{sec:dag-crt-protocol}.

\subsection{Overview of \pname}

\begin{figure}
    \centering
    \includegraphics[scale=0.35]{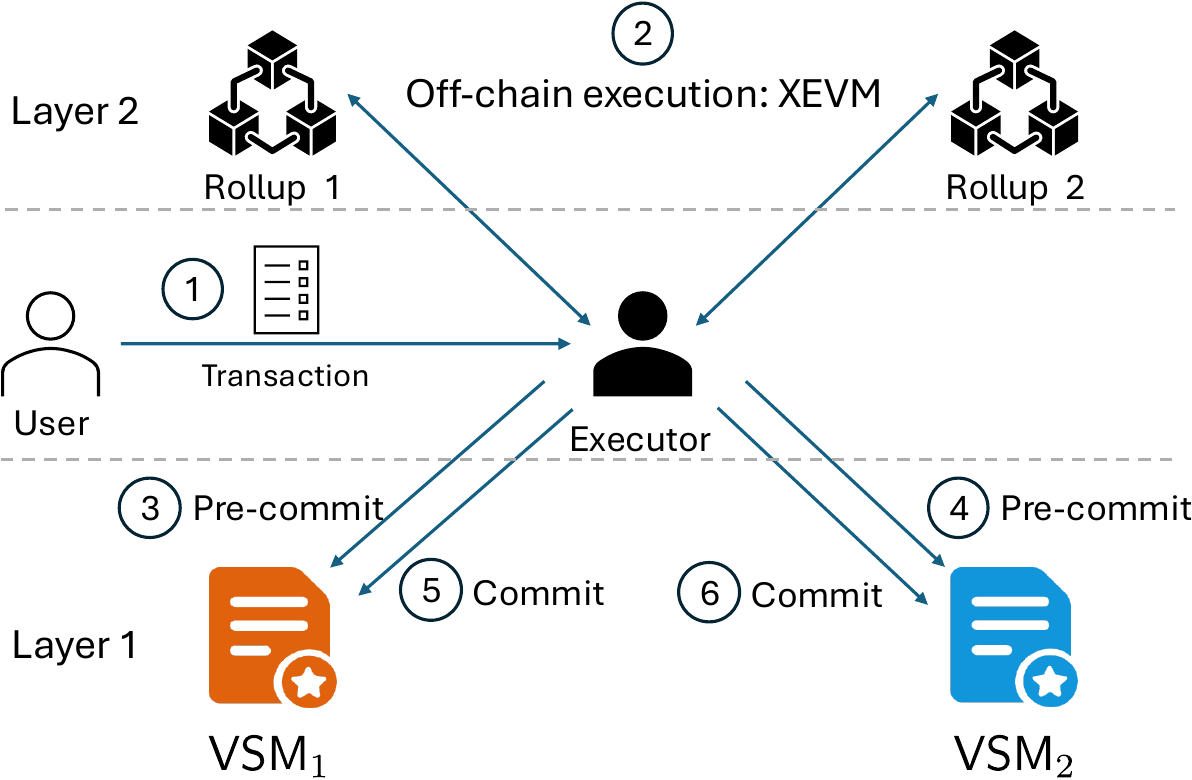}
    \caption{$\pname$ Overview}
    \label{fig:protocol-diagram}
\end{figure}

We describe the end-to-end workflow of $\pname$, illustrated in~\cref{fig:protocol-diagram}, starting from the user CRT and all the way until the pair of new state digests are accepted by the VSM contracts. The user starts by creating a single transaction which encodes her intended CRT by making the appropriate calls to $\gsc.\trig(\cdot)$, and sends her transaction to the Executor (step $1$). The Executor keeps executing received transactions off-chain (phase 1 below) and, at regular intervals, he submits transaction batches -- along with the state digest they produce -- to the VSM on L1 (phase 2 below).

\parhead{Phase 1.} The Executor executes the transaction batch off-chain following the \emph{\textbf{C}ross-EVM (XEVM)} sub-protocol of $\pname$. The off-chain execution yields a new pair of state digests, and the Executor also produces validity proofs, one per rollup. This phase corresponds to step $2$ in~\cref{fig:protocol-diagram} and addresses the first challenge outlined in~\cref{sec:svs-challenges-and-solutions}. 

\parhead{Phase 2.} The Executor drives the two VSMs to perform the \emph{2PC component} of $\pname$ (steps $3-6$ in~\cref{fig:protocol-diagram}). Each VSM can be in either ``free'' or ``paired'' status, and initially both are in free status.
During the \textit{pre-commit} round, one of the two VSMs is deterministically chosen as the leader; suppose it is $\vsm{1}$. The Executor first drives $\vsm{1}$ to pre-commit (step $3$) by submitting the new state digest, validity proof, and evidence that $\vsm{2}$ is in free status; 
$\vsm{1}$ in turn verifies the proof and the evidence, accepts the digest as temporary, and switches to paired status.
Next, the Executor drives $\vsm{2}$ to pre-commit (step $4$), and $\vsm{2}$ also switches to paired status after verifying that $\vsm{1}$ is in paired status.

During the \textit{commit} round, the Executor first drives $\vsm{1}$ to commit
(step $5$) by submitting evidence that (a) $\vsm{2}$ is in paired status
and (b) the roots of the $\gsc$ contracts match. $\vsm{1}$ in turn verifies the evidence, accepts its temporary state digest as final, and stores its ``commit'' decision under a unique identifier of this 2PC instance. Next, the Executor drives $\vsm{2}$ to commit (step $6$), and $\vsm{2}$ also accepts its temporary state digest as final after verifying that that $\vsm{1}$ has issued a ``commit'' decision.

\myparagraph{Abort case} During the second round of 2PC, $\vsm{1}$ (the leader) can also be driven to abort. After verifying relevant evidence that $\vsm{2}$ has not switched to paired status, $\vsm{1}$ stores an ``abort'' decision (under the 2PC identifier), discards its temporary state digest, and ``rolls back'' to the original state digest. Similarly to the commit case, $\vsm{2}$ will also abort (if ever pre-committed) upon verifying $\vsm{1}$'s abort decision.

\myparagraph{Incorrect Evidence} The workflow above describes an honest execution. If the Executor submits incorrect evidence to the $\vsm{}$ during pre-commit or commit, verification checks will fail, causing the $\vsm{}$ to revert the transaction. In the pre-commit round, the $\vsm{}$ would remain in free status, while in the commit round, it would remain locked in paired status until proper commit/abort evidence is provided. Allowing a $\vsm{}$ to unilaterally abort is insecure, as a malicious Executor could drive one $\vsm{}$ to commit and force the other to abort. This highlights the importance of a leader-based approach, where one $\vsm{}$ decides and the other follows, as well as the non-trivial nature of applying 2PC in this setting.

\begin{remark}[VSM Communication via Bridges]~\label{rem:bridge-assumption}
    When we say a VSM ``verifies evidence'' about the other VSM's state, we mean that it verifies a state membership proof against L1 block headers relayed by a trustless blockchain bridge, such as \cite{xiew-zkbridge-22}. We further assume that a VSM only considers bridge-relayed state digests that are final (2 epochs in Ethereum~\cite{buterin-gasper-2020}), eliminating the possibility of any reorg~\cite{eth-org-reorg}.
\end{remark}

\begin{remark}[Setup]~\label{rem:vsm-setup}
    It is crucial that each VSM be aware of its rollup's GSC contract and of the other VSM. After deployment, we assume a secure, one-time setup phase at the end of which, each VSM has hardcoded the address of its rollup's GSC contract as well as the address of the other VSM. Such setups are common practice and outside this paper's scope.
\end{remark}

\parhead{Application: Cross-Rollup Flash Loan}~\label{subsec:application-xfl}
A cross-rollup flash loan involves the following seven steps:
\noindent\textbf{1. Borrow} the desired amount from a flash loan pool (on $\rol{1}$).
\noindent\textbf{2. Burn} the amount (on $\rol{1}$).
\noindent\textbf{3. Mint} the amount (on $\rol{2}$).
\noindent\textbf{4. Arbitrage} (on $\rol{2}$). 
\noindent\textbf{5. Burn} the post-arbitrage amount (on $\rol{2}$).
\noindent\textbf{5. Mint} the post-arbitrage amount (on $\rol{1}$).
\noindent\textbf{7. Repay} the borrowed amount and retain the profit (on $\rol{1}$). Typically, a Token contract supports burn-and-mint, and a User contract contains the logic for the user’s actions (e.g., arbitrage).

We now discuss why cross-rollup flash loans highlight the need for the more expressive DAG CRT model. In the execution of action $a_1$, which includes steps 1 and 2 on $\rol{1}$, both step 3 (mint) and step 4 (arbitrage) must be triggered for execution on $\rol{2}$. In theory, steps 3 and 4 could be combined into a single action $a_2$, and having $a_1$ trigger $a_2$ would eliminate the need for multiple trigger calls. However, this approach presents a challenge in determining which contract will perform the bundling and trigger $a_2$. Specifically, step 3 relates to the Token contract, while step 4 is handled by the User contract. Consequently, the Token contract should (internally) trigger the mint operation upon burning, and the User contract should trigger the arbitrage, resulting in two trigger calls at the end of $a_2$. The Executor should then pick up both trigger calls, bundle them together, and pass them as arguments to single $\gsc.\act(\cdot)$ invocation. These practical desiderata underscore the necessity for a more flexible and expressive programming model, such as our DAG CRT model.

\section{Formal Foundations}~\label{sec:foundations}
\subsection{Safety}
The SVS trigger-action paradigm~\cite{shared-val-seq-23} serves as a \emph{programming model} for cross-rollup transactions, in addition to providing an elegant atomicity check mechanism. In the simple example where $a_1$ triggers $a_2$, causing the Executor to execute $\act{}(a_2)$, it is as if $a_1$ \emph{passes} control to $a_2$, which in turn is not expected to return to its caller. This programming model is akin -- in spirit -- to the \emph{continuation-passing style (cps)} model~\cite{sussman-scheme-1998, appel-continuations-2007}, as well as to the widely used asynchronous event driven architecture~\cite{jansen-ibm-event-driven-2020, rah-even-driven-microservices-2022, ghaemi-pubsub-blockchain-2021}. 
We first define the simpler chain-CRT model, which is similar to the CRT abstraction in~\cite{lu-atomic-cc-interactions-24}, and then present the more involved DAG-CRT model.

\parhead{Chain CRTs}~\label{sec:chain-crt-model}
We model an \emph{action} as a $3$-tuple $a = (\desc{}, \code{}, \descnext{})$. Here, $\desc{} = (\rol{}, \addr{}, \cdata{})$ is the action's \emph{description}; $\rol{}$ is the target rollup (L2)~\footnote{We assume the $\rol{}$ identifier also determines the underlying L1.}, $\addr$ is the smart contract address, and $\cdata$ contains the function identifier and arguments.
Next, $\code{}$ is the \emph{raw code trace} of the action, i.e., it represents the actual opcode sequence executed, \emph{without} the triggering logic. Finally, $\descnext{}$ is \emph{the description of another action} to be executed, and its format is identical to $\desc{}$. Although not explicit, we assume that $a.\desc{}$ uniquely determines both $a.\code{}$ and $a.\descnext{}$; this correspondence is enforced by EVM rules.~\footnote{Similarly, $(\addr{}, \cdata{})$ determines the EVM opcodes to be executed.} We will overload notation and use $a.\rol{}$ to mean $a.\desc{}.\rol{}$. The field $a.\descnext{}$ can be $\nul$, in which case it is ignored.

\done{This model assumes the original SVS~\cite{shared-val-seq-23} trigger-action paradigm, restricted to allow each transaction to trigger only a single other action.}

\begin{remark}
    Looking ahead, $\descnext{}$ is precisely the field inserted in $\ttree$, while $\desc{}$ is the field inserted in $\atree$.  
\end{remark}

We define a \emph{chain CRT} as a sequence $\ccrt := [a_i]_{i \in [n]}$ which satisfies $a_n.\descnext{} = \nul$ and for all $1 \leq i \leq n-1$, \; $a_i.\descnext{} = a_{i+1}.\desc{} \neq \nul$. Since each action in the sequence is modeled to trigger the next action, we notice that each $a_i$ fully determines the sub-CRT $\ccrt [i\colon \! ]$. This means that the entire CRT $\ccrt$ can be represented by the single source action $a_1$, which is exactly the action the user will submit to the Executor. Even though $a_1$ is a compact representation of $\ccrt$, we still represent CRTs by the full sequence of actions to facilitate analysis.

A transaction \emph{batch} is an ordered list of actions that the Executor submits along with the new state digest on L1. An action can be either \emph{local}, i.e., confined within the rollup, or \emph{non-local}, i.e., part of a CRT. 
We say a \emph{state digest} is \emph{local} if it has been produced by a batch consisting of only local actions. Otherwise we say the state digest is \emph{non-local}. We write $a_1 \underset{\batch{}}{\prec} a_2$ to mean that $a_1$ appears before $a_2$ within $\batch{}$.
We define the following predicate:
\begin{itemize}[leftmargin=*]
    \item $\actexec{}(\dst{}, \dstt{}, \batch{}, a) = \true$ if $a \in \batch{}$ \textbf{and} executing $\batch{}$ causes the state digest to transition from $\dst{}$ to $\dstt{}$.
\end{itemize}

We can now define the atomic execution predicate $\atomexec(\ccrt, \{\rol{j}, \dst{j}, \dstt{j}, \batch{j}\}_j)$, which covers an entire CRT instance $\ccrt := [a_i]_{i \in [n]}$ and is defined with respect to the set of rollups $\rolset = \cup_{i \in [n]} \{ a_i.\rol{} \}$ and their state digest pairs and batches.
This predicate evaluates to $\mathsf{True}$ if and only if (1) occurs or both (2a) and (2b) occur:
\begin{enumerate}[leftmargin=*]
    \item For all $j \in [|\rolset|]$ and for all $i \in [n]$ s.t. $a_i.\rol{} = \rol{j}$, $\actexec{}(\dst{j}, \dstt{j}, \batch{j}, a_i) = \false$.
    \item (a) For all $j \in [|\rolset|]$ and for all $i \in [n]$ s.t. $a_i.\rol{} = \rol{j}$,
            $\actexec{}(\dst{j}, \dstt{j}, \batch{j}, a_i) = \true$.
            
        (b) For all $j \in [|\rolset|]$ and for all $i_1, i_2 \in [n]$ s.t. $a_{i_1}.\rol{} = a_{i_2}.\rol{} = \rol{j}$ and $i_1 < i_2 $, it holds that $a_{i_1} \underset{\batch{}}{\prec} a_{i_2}$.
\end{enumerate}

\noindent We now proceed to formally define atomicity for chain CRTs.

\begin{definition}[CRT Atomicity]~\label{def:forward-atom-chain-crt}   
    Let $\ccrt := [a_i]_{i \in [n]}$ be a CRT, 
    and let $\mathcal{R} = \cup_{i \in [n]} \{ a_i.\rol{} \}$ be the set of associated rollups, for some $n \geq 2$. For each $j \in [|\mathcal{R}|]$, let $\dst{j}$ be the \emph{current} state digest of $\vsm{j}$. Also let $(\dstt{j}, \batch{j})$ be the (state digest, transaction batch) pair which is \emph{immediate next accepted} after running CRAT protocol $\protocol$. We say $\protocol$ satisfies \emph{atomicity} if 
    \begin{align*}
        \atomexec(\ccrt, \{\rol{j}, \dst{j}, \dstt{j}, \batch{j}\}_j) = \mathsf{True}.
    \end{align*}
\end{definition}

\parhead{DAG CRTs}~\label{sec:dag-crt-model}
Similar to chain CRTs, a DAG-CRT action is a tuple $a_i := (\desc{}, \code{}, \descnext{})$. Unlike chain-CRTs, $a.\desc{} = [sa_j.\desc{}]_j$ now is an array of \textbf{s}ub-\textbf{a}ction descriptions (hence ``$sa$''), $a.\code{}$ represents the concatenation of the $[sa_j.\code{}]_j$ components, and $a'.\descnext{} = [\descnext{}[k]]_k $ is an array of descriptions of triggered sub-actions. Similar to chain CRTs, $a.\desc{}$ uniquely determines both $a.\code{}$ and $a.\descnext{}$. The array $\descnext{}$ can be empty, in which case it is ignored. We define a DAG-CRT as a sequence $\dagcrt := [a_i]_{i \in [n]}$ which satisfies $|a_n.\descnext{}| = 0$ and for all $1 \leq i \leq n-1$, $a_i.\descnext{} = a_{i+1}.\desc{}$. Note the equality is between \emph{arrays}.

We can now repurpose all predicates from~\cref{sec:chain-crt-model} for DAG CRTs. Under the repurposed predicate $\atomexec{}(\dagcrt, \cdot)$), we can also re-purpose~\cref{def:forward-atom-chain-crt} to work for DAG CRTs. We omit re-stating these repurposed versions since the only difference is the use of $\dagcrt$ instead of $\ccrt$.

\subsection{Liveness and Efficiency}

\myparagraph{Liveness} A CRAT protocol satisfies liveness if the participating rollups maintain liveness, as discussed in \cref{sec:problem-statement}.

\myparagraph{Efficiency}
The key efficiency metric for cross-rollup transactions (CRTs) is latency, which refers to the time delay between submitting a CRT instance and its successful commitment or abortion. Measuring time complexity in this setting is challenging due to the heterogeneous nature of participating rollups and their differing finality rules. However, L1 transactions are the bottleneck, and we define a "round" as the duration of the slowest layer 1 transaction, based on the blockchain that takes the longest to finalize a transaction. CRAT protocols should preserve the constant round finality of existing rollups, and we capture this requirement in~\cref{def:efficiency}.

\begin{definition}[Efficiency]~\label{def:efficiency}
    Let $\protocol$ be a CRAT protocol and let $\ccrt = [a_i]_{i \in [n]}$ be a chain-CRT of length $n$. We say $\protocol$ is \emph{efficient} if it completes in $O(1)$ rounds.
\end{definition}

Similarly, we can repurpose \cref{def:efficiency} for DAG CRTs by requiring that the protocol completes in $O(1)$ rounds regardless of the CRT length (number of actions) as well as the total number of sub-actions.

\done%

\section{\pname: Off-Chain Execution -- XEVM}~\label{sec:xevm}
\newcommand{\startsession}{\textproc{startSession}}
\newcommand{\checksessionid}{\textproc{checkSessionID}}
\newcommand{\trigger}{\textproc{trigger}}
\newcommand{\triggerinternal}{\textproc{triggerInternal}}
\newcommand{\action}{\textproc{action}}

In this section, we present XEVM, our modified execution engine for rollups. First, we extend the General System Contract ($\gsc$) introduced in~\cite{shared-val-seq-23} so that it securely supports chain CRTs of arbitrary length. Next, we present the full Executor algorithm specification for the XEVM component.

\subsection{General System Contract}

\begin{algorithm}[ht]
    \small
    \caption{General System Contract}
    \label{alg:gsc-single-trigger}
    \begin{algorithmic}[1]

    \Function{\textcolor{chain}{startSession}}{\textcolor{chain}{$\addr{}, \cdata$}}
            \State \textcolor{chain}{ \require ($\Not \; \acalled{}$) }
            \State \textcolor{chain}{ $\enonce{} \gets \enonce{} + 1$ }
            \State \textcolor{chain}{ $\snonce{} \gets \snonce{} + 1$ }
            \State \textcolor{chain}{ $\sactive{} \gets \true$ }
            \State \textcolor{chain}{ $\tcalled{} \gets \false$ }
            \State \textcolor{chain}{ \Call{triggerInternal}{$\msgsender, \addr{}, \cdata$} }
    \EndFunction
    \Function{\textcolor{chain}{checkSessionID}}{\textcolor{chain}{$\sid{}$}}
        \If{\textcolor{chain}{ $\sid{} = \snonce{}$ } } 
            \State \textcolor{chain}{ \require ($\sactive{}$) }
        \ElsIf{\textcolor{chain}{ $\sid{} = \snonce{} + 1$} }
            \State \textcolor{chain}{ $\sactive{} \gets \true$ }
            \State \textcolor{chain}{ $\snonce{} \gets \snonce{} + 1$ }
        \Else
            \; \textcolor{chain}{ \textbf{Revert} }
        \EndIf
    \EndFunction
    \Function{triggerInternal}{$\sender, \addr{}, \cdata$}
        \State  \textcolor{chain}{ \require ($\Not \; \tcalled{}$) \textcolor{comment}{\Comment{Only one trigger per tx}} }
        \State \textcolor{chain}{$\tcalled{} \gets \true$}
        \State \textcolor{chain}{$\sid{} \gets \snonce{}$}
        \State $\tnonce{} \gets \tnonce{} + 1$
        \State $ \msghash \gets \hash(\sender, \addr{}, \cdata, \tnonce{}, \textcolor{chain}{\sid{}})$
        \State \textbf{emit T}($\sender, \addr{}, \cdata, \tnonce{}$)
        \State $\ttree.\insertt(\msghash)$
    \EndFunction
    \Function{trigger}{$\addr{}, \cdata$}
        \State \textcolor{chain}{ \Call{triggerInternal}{$\msgsender, \addr{}, \cdata$} }
    \EndFunction
    \Function{action}{$\sender, \addr{}, \cdata, \textcolor{chain}{\sid{}}$}
        \State \textcolor{chain}{$\acalled{} \gets \true$}
        \State \textcolor{chain}{\Call{checkSessionId}{$\sid{}$}}
        \State \textcolor{chain}{$\tcalled{} \gets \false$}
        \State (status, \_) $\gets \addr{}.\call(\cdata)$ \textcolor{comment}{\Comment{execute triggered action} }
        \State \require (status)
        \State $\anonce{} \gets \anonce{} + 1$
        \State $ \msghash \gets \hash(\sender, \addr{}, \cdata, \anonce{}, \textcolor{chain}{\sid{}}) $
        \State $\atree.\insertt(\msghash)$
        \If{\textcolor{chain}{ ($\Not \; \tcalled{}$) } }
            \textcolor{chain}{ $\sactive{} \gets \false$ }
        \EndIf
        \State \textcolor{chain}{$\acalled{} \gets \false$}
        
    \EndFunction
    \end{algorithmic}
\end{algorithm}

 In this section, we present our modified General System Contract ($\gsc$) which handles chain CRTs. \cref{alg:gsc-single-trigger} fully specifies our $\gsc$ contract; we highlight our contributions with \textcolor{chain}{blue color} on top of SVS~\cite{shared-val-seq-23}. Besides the Merkle trees $(\ttree, \atree)$ and nonces ($\tnonce{}, \anonce{})$, the $\gsc$ contract maintains two additional nonces $(\snonce{}, \enonce{})$ and the $\sactive{}$ flag to keep track of CRT sessions. Nonces are initialized to zero and the $\sactive{}$ flag to $\false$.

The source action $(a_1)$ launches a new CRT session by calling the new entry point function $\startsession$, passing the $(\addr, \cdata)$ fields of the action it wishes to trigger. $\startsession$ increments both $\enonce{}$ and $\snonce{}$ by one, marks the session as active ($\sactive{} \gets \true$), and calls $\triggerinternal(\msgsender, \addr, \cdata)$. $\triggerinternal$ is an internal function only supposed to be called by $\startsession$. 
$\trigger$ is the function that application smart contracts are supposed to invoke.

The $\trigger$ function increments $\tnonce{}$, hashes the tuple $(\msgsender, \addr{}, \cdata, \tnonce{}, \sid{})$, and inserts the hash into $\ttree$. Here $\msgsender$ is the address of the caller contract and $\sid{}$ is taken to be the current $\snonce{}$ value. It also emits a Trigger event intended to be listened to by the Executor.  

The $\action$ function takes as input the tuple $(\sender, \addr{}, \cdata, \sid{})$ and can only be called by the Executor. First, it checks that $\sid{}$ matches either $\snonce{}$, which requires that $\sactive{} = \true$ , or $\snonce{} + 1$, signaling a new CRT session. Then, it sets $\xsender{} \gets \sender$ and calls the smart contract method specified by $(\addr{}, \cdata)$. The callee contract can access the public variable $\xsender{}$ which specifies the cross-rollup sender. Next, it increments $\anonce{}$, hashes the tuple $(\sender, \addr{}, \cdata, \anonce{}, \sid{})$, and inserts the hash into $\atree$. If the $\tcalled{}$ flag is raised during the call to $(\addr{}, \cdata)$ (i.e., there is a nested call to $\trigger$), this means that the current CRT session is not finished. Otherwise, this is the end of the current CRT session and it marks $\sactive{} \gets \false$.

We assume a unique copy of the $\gsc$ contract is deployed on each rollup. Our design guarantees that the $\trigger$ function can only be called either by $\startsession$ or in a nested manner during an $\action$ call, as well as that $\startsession$ can only be called by a source action ($a_1$), and never in a nested manner during an $\action$ call.

\subsection{Executor Specification}

In this section, we describe the XEVM component of the Executor, formally specified in~\cref{alg:shared-executor-xevm-chain-crt}. The Executor algorithm manages the execution of both local and non-local transactions across the two rollups. We assume a function $\textproc{evm}$ which takes as input a transaction and executes it on the specified rollup according to EVM rules. The Executor maintains the current state and a state checkpoint for each rollup. To execute a transaction $\tx$ received by a user, he calls $\textproc{entryPoint}(\tx)$, which creates state checkpoints for all rollups before proceeding to call \textproc{xevm}($\tx$). The $\textproc{xevm}(\tx)$ function uses the (assumed) $\textproc{evm}$ function to execute $\tx$ and listens for a trigger event. If there is a triggered transaction $\tx'$, the Executor constructs a wrapper transaction $\tx''$ which calls $\gsc.\action(\cdot)$ with the description of $\tx'$ as argument. It then recursively calls $\textproc{xevm}(\tx'')$, and so on, until all trigger calls have been handled. If any transaction or its triggered transaction fail, this failure propagates up to $\textproc{entryPoint} $, which restores the state of all rollups to their checkpoints.

\section{\pname: 2PC Protocol}~\label{sec:2pc-trustless}
In this section, we describe the Two-Phase Commit (2PC) protocol between the two $\vsm{}$s and driven by the Executor. We omit formally specifying the Executor's role in the 2PC protocol, since it follows directly from the $\vsm{}$ specification. As such, the entire section focuses on the $\vsm{}$ specification.

The $\vsm{}$ contract consists of three main functions ($\textproc{updateDigest}$, $\textproc{commit}$, and $\textproc{abort}$), shown in \cref{alg:vsm-generic}, which serves as a generic 2PC framework between VSMs. Concrete procedures -- such as leader determination and evidence verification -- are performed by functions $\textproc{IsLocal}, \textproc{VerPreComEvd}, \textproc{VerComEvd}$, and $\textproc{VerAbEvd}$, which we describe in this section and formally specify in \cref{alg:vsm-trustless-p1} in~\cref{apdx:vsm}.

\myparagraph{Bridge abstraction}
Recall from \cref{par:background-bridges} and \cref{rem:bridge-assumption} that we assume a trustless bridge between the underlying L1 chains. In the following, the $\ovsm$ structure contains state attributes of the other $\vsm{}$, which can be verified against the bridge using attribute proofs which we denote by $\attproofs$. 

\begin{algorithm}[ht]
    \small
    \caption{Validator Smart Contract: Generic 2PC}
    \label{alg:vsm-generic}
    \begin{algorithmic}[1] %

        \Function{UpdateDigest}{$\dstt{},  \valproof, \locevd, \pcevd$}
            \State \textbf{Require} $(\pstat = \free)$
            \State \Call{VerValProof}{$\dstt{}, \valproof$}
            \If{\Call{IsLocal}{$\dstt{}, \locevd$}} \; $\dst{} \gets \dstt{}$
            \Else
                \State \Call{VerPreComEvd}{$\dstt{}, \pcevd$}
                \State $\dstemp{} \gets \dstt{}$, \; $\pstat \gets \paired$
            \EndIf
        \EndFunction
        \Function{Commit}{$\cevd$}
            \State \require $(\pstat = \paired)$
            \State \Call{VerComEvd}{$\cevd$}
            \State $\dst{} \gets \dstemp{}$, \; $\pstat \gets \free$
        \EndFunction
        \Function{Abort}{$\abevd$}
            \State \require $(\pstat = \paired)$
            \State \Call{VerAbEvd}{$\abevd$}
            \State $\pstat \gets \free$
        \EndFunction
    \end{algorithmic}
\end{algorithm}

\myparagraph{State} The $\vsm{}$'s state consists of its ID ($\vsmid{}$) and the current rollup state digest ($\dst{}$). We also assume that it knows the address and ID of the VSM that it is connected to (see \cref{rem:vsm-setup}). The $\pstat \in \{\free, \paired \}$ variable is initialized to $\free$. All remaining state variables are specified in~\cref{alg:vsm-trustless-p1}.

\myparagraph{Entrypoint} The $\textproc{updateDigest}$ function takes as input a new state digest ($\dstt{}$), a validity proof ($\valproof$), local evidence ($\locevd$), and pre-commit evidence. First, the $\vsm{}$ verifies $\valproof$ using the (assumed) existing zk-proof verification mechanism (\textproc{VerValProof} function). Next, it parses the local evidence $(\roots, \attproofs) \gets \locevd$ which is used to determine whether $\dstt{}$ is a local or non-local state digest. It parses the claimed roots of the $\gsc$ contract  $(\troot{}, \aroot{},  \troot{}', \aroot{}') \gets \roots$ and verifies $\attproofs$ to check the membership of $(\troot{}, \aroot{})$ (resp. $(\troot{}', \aroot{}')$) against $\dst{}$ (resp. $\dstt{}$). If $\troot{}' = \troot{}$ and $\aroot{}' = \aroot{}$, then $\dstt{}$ is a $\emph{local}$ digest and thus accepted immediately $(\dst{} \gets \dstt{})$. Otherwise, $\dstt{}$ is a $\emph{non-local}$ digest and the $\vsm{}$ proceeds to pre-commit. Note that $\locevd$ additionally contains the $\snonce{}'$ and $\enonce{}'$ variables of the $\gsc$ contract, which are also verified under $\dstt{}$, and stored to be used in the commit round. 

\myparagraph{Pre-commit} If the $\vsm{}$ successfully verifies the supplied pre-commit evidence ($\pcevd$), then it stores $\dstt{}$ as the temporary digest and switches to $\paired$ status.
 In the $\textproc{VerPreComEvd}$ function, the $\vsm{}$ parses the pre-commit evidence $(\ovsm, \attproofs, \dstt{O}) \gets \pcevd$, where $(\ovsm, \attproofs)$ are verified against a bridge-relayed state digest of the other L1 chain. Next, it computes the 2PC instance identifier as $\idx \gets \textproc{GetIndex}(\dst{}, \dstt{}', \ovsm.\dst{}, \dstt{O})$, where $\textproc{GetIndex}(\cdot)$ deterministically orders and hashes the four state digests, so that both VSMs compute the same $\idx$ value. The $\idx$ identifier is stored and used to determine the leader of this 2PC instance. If the $\vsm{}$ is the \textit{leader}, it checks whether $\ovsm.\pstat = \free$. If the $\vsm{}$ is \textit{not the leader}, it checks that the computed index matches the leader VSM's index and that the leader VSM is in $\paired$ status.

\myparagraph{Commit/Abort} If the $\vsm{}$ successfully verifies the supplied commit evidence ($\cevd$), then it accepts the temporary state digest as final ($\dst{} \gets \dstemp{}$) and switches to $\free$ status.
 In the $\textproc{VerComEvd}$ function, the $\vsm{}$ parses the commit evidence $(\ovsm, \attproofs) \gets \cevd$, where $\ovsm$ and $\attproofs$ are first verified similarly to the pre-commit round. Next, it fetches the current 2PC instance identifier $\idx$ and determines the leader of this 2PC instance. If the $\vsm{}$ is the \textit{leader}, it checks whether $\ovsm.\pstat = \paired$ and $\ovsm.\idx = \idx$. It also verifies that the new pair of state digests ($\dstemp{}, \dst{}$) satisfy \emph{atomicity} by checking that (1) the Merkle roots of the $\gsc$ contracts match (i.e., $\troot{}' = \ovsm.\aroot{}'$ and $\aroot{}' = \ovsm.\troot{}'$) and (2) the sum of the two entry nonces is equal to both session nonces. If so, it stores a $\commit$ decision under $\idx$. If the $\vsm{}$ is \textit{not the leader}, it follows the leader VSM's decision under $\idx$. The $\textproc{abort}$ and $\textproc{verAbEvd}$ functions are similar to the commit case; see~\cref{alg:vsm-trustless-p1} for more details.

\section{\pname: Safety and Liveness}~\label{sec:safety-liveness}
\subsection{Safety}~\label{sec:safety-chain-crt}

In this section, we prove that $\pname$ satisfies atomicity per~\cref{def:forward-atom-chain-crt}. First, we prove that if the off-chain execution of the CRT does not satisfy all-or-nothing or serializability, then certain predicates about the two $\gsc$ states will not hold. Second, we prove that if at least one of the $\gsc$ state predicates does not hold, then the two $\vsm{}$s on L1 will not be able to complete the 2PC and commit to the new pair of faulty state digests. We use three helper lemmas which collectively imply the first step of our proof sketch. Note that \cref{lem:svs-gsc-all-or-none-exec} and \cref{lem:svs-gsc-honest-or-reverse-order} assume the SVS~\cite{shared-val-seq-23} $\gsc$ contract. Our modifications to the $\gsc$ contract do not affect these lemmas, so their implications directly apply to the more robust chain-CRT $\gsc$.

 An observation that will prove useful is that if a non-local action $a$ is part of an L1-accepted batch, then the fields $a.\desc{}$ and $a.\descnext{}$ are inserted into the trigger and action trees, respectively -- besides $a.\code{}$ being executed. We capture this intuitive obesrvation in \cref{claim:a-exec-iff-desc-inserted-chain}. 
 
\begin{claim}~\label{claim:a-exec-iff-desc-inserted-chain}
    Let $\ttree$ and $\atree$ be the trigger and action trees of $\gsc$ after $\batch{}$ is accepted by the $\vsm{}$ of a rollup, and let $a$ be a \emph{non-local} action. Then $a \in \batch{}$ if and only if $a.\desc{} \in \atree$ and $a.\descnext{} \in \ttree$.~\footnote{This also captures the edge cases when $a$ is the first or last CRT action; e.g. if $a.\desc{} = \nul$, then trivially $a.\desc{} \in \atree$.}
\end{claim}

\begin{proof}
     $(a.\desc{}, a.\code{}, a.\descnext{})$ comprise the same action $a$ and are inserted/executed atomically when $a \in \batch{}$.
\end{proof}

\cref{lem:svs-gsc-all-or-none-exec} states that the SVS $\gsc$ satisfies the all-or-nothing property of atomicity.

\begin{lemma}~\label{lem:svs-gsc-all-or-none-exec}
    Let $\ccrt = [a_i]_{i \in [n]}$ and $\{\rol{b}, \dst{b}, \dstt{b}, \batch{b}\}_{b \in \{1,2\}}$ be as in~\cref{def:forward-atom-chain-crt}, when using the \emph{original SVS} $\gsc$ contract for off-chain execution. For $b \in \{1,2\}$, let $\{\ttree_b, \atree_b \}$ be the $\gsc_b$ trees under $\dst{b}$.
    Also assume (I) $\ttree_1 = \atree_2$ and $\ttree_2 = \atree_1$.
    Then (I) is preserved under $(\dstt{1}, \dstt{2})$ if and only if there exists $d \in \{\true, \false\}$ such that, for all $i \in [n]$, $\actexec{}(\dst{b}, \dstt{b}, \batch{b}, a_i) = d$ for both $b \in \{1,2\}$ where $a_i.\rol{} = \rol{b}$.
\end{lemma}

\begin{proof}[Proof of \cref{lem:svs-gsc-all-or-none-exec}]
    The ``if'' direction corresponds to the honest execution; its proof is simple and thus omitted. We prove the contrapositive of the ``only if'' direction. Assume there exists $1 \leq i \leq n-1$ and such that (without loss of generality)  $a_{i}'.\rol{} = \rol{1}$ and $a_{i + 1}'.\rol{} = \rol{2}$, and $\actexec{}(\dst{1}, \dstt{1}, \batch{1}, a_{i}') = d_i \neq d_{i+1} = \actexec{}(\dst{2}, \dstt{2}, \batch{2}, a_{i+1}')$. Since $\ccrt$ is a chain CRT, we know $a_i'.\descnext{} = a_{i+1}'.\desc{} \neq \nul$.
    
    \noindent\textbf{Case 1:} $d_i = \true$ and $d_{i+1} = \false$. Using \cref{claim:a-exec-iff-desc-inserted-chain}, the former implies $a_{i}'.\descnext{} \in \ttree_{1}'$, and the latter implies $a_{i+1}'.\desc{} \notin \atree_{2}'$.
            
    \noindent\textbf{Case 2:} $d_i = \false$ and $d_{i+1} = \true$. The former implies $a_{i}'.\descnext{} \notin \ttree_{1}'$, and the latter implies $a_{i+1}'.\desc{} \in \atree_{2}' $.
    In either case we get $\ttree_{1}' \neq \atree_{2}'$.
\end{proof}

\cref{lem:svs-gsc-honest-or-reverse-order} states that under the SVS $\gsc$, an Executor can only \emph{fully reverse} the order of actions within a CRT -- and still pass the $\gsc$ root check. We defer its proof to \cref{apdx:chain-crt-proofs}. In the following, we say that $\batch{}$ \emph{preserves} (resp. \emph{fully reverses}) the relative order of actions in a CRT $\ccrt = [a_i]_{i \in [n]}$ if for all $0 \leq i_1 < i_2 \leq n$ such that $a_{i_1}, a_{i_2} \in \batch{}$, it holds that $a_{i_1} \underset{\batch{}}{\prec} a_{i_2}$  (resp. $a_{i_1} \underset{\batch{}}{\succ} a_{i_2}$).

\begin{lemma}~\label{lem:svs-gsc-honest-or-reverse-order}
    Consider the same setup as in \cref{lem:svs-gsc-all-or-none-exec} and assume $\ttree_1 = \atree_2$ and $\ttree_2 = \atree_1$. Also assume $\actexec{}(\dst{b}, \dstt{b}, \batch{b}, a_i) = \true$ for $b \in \{1,2\}$ where $a_i.\rol{} = \rol{b}$. Then both $\batch{1}$ and $\batch{2}$ either preserve or fully reverse the relative order of actions if and only if $\ttree_1' = \atree_2'$ and $\ttree_2' = \atree_1'$.
\end{lemma}

\cref{lem:chain-gsc-honest-order} states that our chain-CRT $\gsc$ is safe against the ``reversal'' attack introduced in \cref{lem:svs-gsc-honest-or-reverse-order}. We defer the proof of \cref{lem:chain-gsc-honest-order} to~\cref{apdx:chain-crt-proofs}. 

\begin{lemma}~\label{lem:chain-gsc-honest-order}
    Let $\ccrt = [a_i]_{i \in [n]}$ and $\{\rol{b}, \dst{b}, \dstt{b}, \batch{b}\}_{b \in \{1,2\}}$ be as in \cref{def:forward-atom-chain-crt}, when using our \emph{chain-CRT} $\gsc$ contract for off-chain execution. For $b \in \{1,2\}$, let $\{\ttree_b, \atree_b, \sactive{b}, \snonce{b},  \\ \enonce{b} \}$ be the $\gsc_b$ state under $\dst{b}$. Assume $\actexec{}(\dst{b}, \dstt{b}, \batch{b}, a_i) = \true$ for $b \in \{1,2\}$ where $a_i.\rol{} = \rol{b}$. Also assume 

    \noindent(I) $\ttree_1 = \atree_2$ and $\ttree_2 = \atree_1$,
    
    \noindent(II) $\snonce{1} = \snonce{2} = \enonce{1} + \enonce{2}$, and

    \noindent(III) $\sactive{b} = \false$ for at least one $b \in \{1,2\}$.

    \noindent Then $\batch{1}$ and $\batch{2}$ preserve the relative order of all actions if and only if (I,II,III) are preserved under $(\dstt{1}, \dstt{2})$.
\end{lemma}

We now state and prove our main theorem.

\begin{theorem}~\label{thm:chain-crate-forward-atomicity-two-rollups}
    $\pname$ satisfies atomicity for chain-CRTs and $|\mathcal{R}| = 2$ rollups.
\end{theorem}

\begin{proof}[Proof of \cref{thm:chain-crate-forward-atomicity-two-rollups}]
    Let $\ccrt = [a_i']_{i \in [n]}$ and $\{\vsm{b}, \rol{b}, \dst{b}, \dstt{b}, \batch{b}\}_{b \in \{1,2\}}$ be as in \cref{def:forward-atom-chain-crt}, where $\{\dstt{b}, \batch{b}\}_{b \in \{1,2\}}$ are the new state digests and batches after running $\pname$. For $b \in \{1,2\}$, let $\{\ttree_b, \atree_b, \sactive{b}, \snonce{b}, \\ \enonce{b} \}$ be the $\gsc_b$ state under $\dst{b}$.
    Assume properties (I,II,II) from \cref{lem:chain-gsc-honest-order} all hold under $(\dst{1}, \dst{2})$.
    Also let $\idx = \textproc{getIndex}(\dst{1}, \dstt{1}, \dst{2}, \dstt{2})$ be the 2PC instance index.
    For contradiction, assume $\pname$ does \emph{not} satisfy atomicity, i.e., assume
    \begin{align*}
        \atomexec(\ccrt, \{\rol{b}, \dst{b}, \dstt{b}, \batch{b}\}_{b \in \{1,2\}}) = \false.
    \end{align*}
    This means either (a) the $\actexec{}(\cdot)$ predicates do not all agree, or (b) all $\actexec{}(\cdot) = \true$ and at least one pair of actions are not in order in their batch. In case (a), \cref{lem:svs-gsc-all-or-none-exec} says that property (I) is not preserved. In case (b), \cref{lem:chain-gsc-honest-order} says that at least one of (I,II,III) is not preserved. 
    In either case, using an inductive argument across all CRTs present in $\batch{1}$ and $\batch{2}$, at least one of (I,II,III) is not preserved under $(\dstt{1}, \dstt{2})$.
    
    \smallskip
    We now consider the 2PC component of $\pname$ leading to the new state digests. First note that in either case (a) or (b), at least one of $(\dstt{1}, \dstt{2})$ is a non-local digest; without loss of generality, assume $\dstt{1}$ is non-local. Thus $\vsm{1}$ must have executed $\textproc{PreCommit}$ successfully, followed by $\textproc{Commit}$, in order to commit to the new state digest $\dstt{1}$. This is because $\vsm{1}$ will see that one of $\ttree_1', \atree_1'$ is updated and will never accept $\dstt{1}$ as a new local state digest. This relies on the soundness of merkle membership proofs of the $\gsc$ state against the rollup state digest. We consider two cases:
    
   \noindent\textbf{Case 1:} $\vsm{1}$ is the leader under $\idx$. Then $\vsm{2}$ must have successfully  executed $\textproc{PreCommit}$ and entered $\paired$ status under $\idx$ 
    (otherwise, $\vsm{1}$ will never see correct proof of $\vsm{2}$ being $\paired$, and will never decide to commit). Now consider the successful execution of $\textproc{Commit}$ on $\vsm{1}$, 
    leading to a $\commit$ decision. Since $\vsm{1}$ decides to commit, it must see that all three properties (I,II,III) are preserved under $(\dstt{1}, \dstt{2})$, a contradiction. 

    \noindent\textbf{Case 2:} $\vsm{1}$ is \emph{not} the leader under $\idx$, so $\vsm{2}$ is the leader. Since $\vsm{1}$ accepts new non-local state digest ($\dstt{1}$), it must accept it after successful completion of $\textproc{Commit}$. Since $\vsm{1}$ is not the leader, it must have seen correct proof that $\vsm{2}.\decisions[\idx] = \commit$. For $\vsm{2}$ to decide to commit, it must have also successfully executed $\textproc{Commit}$ and committed. Again, it must be that all properties (I,II,III) are preserved under $(\dstt{1}, \dstt{2})$, a contradiction.  

\end{proof}

Note that the proof of \cref{thm:chain-crate-forward-atomicity-two-rollups} relies on the soundness of bridge proofs for the correctness of the L1 state digest.

\subsection{Latency and Liveness}~\label{sec:liveness-chain-crt}
$\pname$ requires two L1 transactions per rollup in sequential order: $\vsm{1}.\textproc{updateDigest}$, $\vsm{2}.\textproc{updateDigest}$, $\vsm{1}.\textproc{Commit}$, and $\vsm{2}.\textproc{Commit}$, with a similar order for aborts. This ensures an end-to-end latency of $4$ rounds, satisfying efficiency with respect to \cref{def:efficiency}. $\pname$ satisfies \emph{liveness}, as the leader VSM can always exit $\paired$ status by committing or aborting, while the non-leader VSM depends on the leader’s decision. Since both rollups are assumed live, this dependency does not cause deadlocks.

\section{\pname: Supporting DAG CRTs}~\label{sec:dag-crt-protocol}
In this section, we
augment $\pname$ to support the more expressive DAG CRT programming model. We only modify the XEVM component of CRATE and leave the 2PC component on L1 unchanged; as a result, the liveness argument is identical to the chain-CRT version.

\parhead{General System Contract}
The chain-CRT-based $\gsc$ contract does not support the more expressive DAG CRTs, since it crucially relies on marking the session as inactive upon encountering a non-triggering action.
To support DAG CRTs, we modify XEVM as follows. 

First, $\action(\cdot)$ now accepts an \emph{array} of triggered actions, instead of a single action. Correspondingly, we modify $\trigger(\cdot)$ to accept trigger calls even after the $\tcalled{}$ flag is set to $\true$. 
To force Executors to use a single $\action$ call, we insert logic so that all trigger (resp. action) $\msghash$ values created as part of the same parent transaction use the same $\tnonce{}$ (resp. $\anonce{}$) value. At this point, notice that in the chain-CRT $\gsc$, the $\textproc{startSession}$ function is designed to route a single $\trigger$ call, thus limiting $a_1$ to single trigger call. To allow $a_1$ to make multiple trigger calls (which can originate from distinct contracts, such as in a cross-rollup flash loan), we re-purpose $\textproc{startSession}$ to be the \emph{entrypoint} of the DAG CRT. We require that the user route her compiled $a_1'$ action through $\textproc{startSession}$ by passing $(a_1'.\addr{}, a_1'.\cdata{})$ as arguments to it. $\textproc{startSession}$ then executes $a_1'$ and the $\gsc$ handles all nested $\trigger$ calls as described above. We leave $\textproc{checkSessionID}$ and all session-tracking logic unchanged.
We present our full DAG-CRT $\gsc$ contract in \cref{apdx:gsc}.

\parhead{Executor specification}
The Executor now listens for multiple trigger events, and passes the array all triggered actions as input to $\gsc.\action(\cdot)$. We reflect this change in the modified $\textproc{xevm}$ function in \cref{alg:shared-executor-xevm-dag}.

\parhead{Safety}
We now state our second main theorem, which says that $\pname$ is secure with respect to \cref{def:forward-atom-chain-crt} repurposed for DAG-CRTs. We defer the theorem's proof to \cref{apdx:dag-crt-proofs}.

\begin{theorem}~\label{thm:dag-crate-forward-atomicity-two-rollups}
     $\pname$ satisfies atomicity for DAG CRTs and $|\mathcal{R}| = 2$ rollups.
\end{theorem}

\section{Implementation and Evaluation}~\label{sec:implementation}

We implement $\pname$ on Ethereum networks using the Foundry toolkit~\cite{foundry-book}. We abstract away the (existing) SNARK validity proofs of zk-rollups and implement the additional logic introduced by $\pname$. We also implement our motivating application, the cross-rollup flash loan, on top of $\pname$. %

\subsection{Implementation details}

Our $\pname$ implementation~\footnote{\url{https://github.com/ikaklamanis/crate}} consists of three main components: the Executor which performs both the XEVM off-chain execution on L2 and drives the 2PC on L1 (implemented in 1,200+ lines of Python), the L2-based General System ($\gsc$) contract (implemented in 100+ lines of Solidity), and the L1-based Validator ($\vsm{}$) contract (implemented in 500+ lines of Solidity). We conducted our evaluation on a machine with 80 cores powered by an Intel(R) Xeon(R) Platinum 8380 CPU @ 2.30GHz, 125 GB of memory, and 208 GB of swap space.

\myparagraph{Rollup architecture}
We use a simple rollup architecture with two rollups. Each rollup consists of an L2 chain instantiated as a Foundry~\cite{foundry-book} EVM chain, whose state is checkpointed on an L1 chain, also instantiated as a Foundry EVM chain. We deploy a copy of the $\gsc$ contract on the L2 chain, and a copy of the $\vsm{}$ contract on the L1 chain. 

\myparagraph{Executor}
We implement the Executor in Python as a standalone entity consisting of two modules. First, the XEVM module implements the Shared Executor \cref{alg:shared-executor-xevm-chain-crt}. 
Second, the 2PC module implements the Executor's role in the 2PC protocol; it fetches the state roots of both rollups, as well as membership proofs for the $\gsc$ state, using the Ethereum RPC API~\cite{eth-rpc-api}. 
We simulate the generation of the user transaction by including it as part of the Executor implementation. The Executor also simulates the bridge between the two L1s, by relaying block headers to the two VSM contracts. 

\myparagraph{$\vsm{}$ contract}
The $\vsm{}$ contract receives new state digests and follows the 2PC protocol to commit or abort. Our implementation executes the full 2PC protocol, which includes verifying MPT membership proofs that are part of the pre-commit and commit evidence. To optimize gas efficiency, we batch membership proofs where possible. We provide and evaluate two implementations for MPT verification. In the first implementation (``CRATE-MPT''), we use a smart contract~\cite{solidity-mpt} which directly verifies membership proofs on-chain; in the second one (``CRATE-SNARK''), we use Groth16~\cite{groth-16} SNARKs  and the Circom~\cite{circom-site} framework. %

\newcommand{\token}{\mathsf{xToken}}
\newcommand{\flpool}{\mathsf{xFlashLoan}}
\newcommand{\userfl}{\mathsf{xUserFL}}

\parhead{Flash Loan Application}
We implement a flexible cross-rollup flash loan infrastructure in less than 300 lines of Solidity, carefully leveraging the more expressive DAG-CRT model of $\pname$. The implementation consists of two ``service'' contracts: the $\token$ contract, which facilitates token transfers and the \emph{cross-rollup} burn-then-mint service, and the $\flpool$ contract, which supports the \emph{cross-rollup} flash loan service.
Third, we implement the $\userfl$ contract, which is executes the user-desired actions during the cross-rollup flash loan. We deploy a copy of each of the three contracts on each rollup L2 chain.
The reader can refer to \cref{apdx:xfl} for the Solidity implementation of these contracts.

\subsection{Evaluation}

\newcommand{\batchsize}{\mathsf{batchSize}}
\newcommand{\numbatches}{\mathsf{numBatches}}

\parhead{L1 Gas Usage}

Compared to existing zk-rollups, $\pname$ introduces an extra L1 transaction for finalizing the state digest and additional overhead for verifying pre-commit and commit evidence. We measure this extra L1 gas usage and compare it to Zksync Era~\cite{zksync-era}, which finalizes an L2 batch using three L1 transactions~\cite{quarkslab-zksync-workflow}. Based on data from Etherscan~\cite{etherscan} and Zksync Explorer~\cite{zksync-explorer}, the L1 gas usage of Zksync is in the range $0.9$M -- $2.1$M gas during the period June-December 2024. Since zk-rollups already incur batch submission costs, $\pname$’s evaluation excludes batch size considerations.

We define a simple workflow where the Executor receives transactions, which are grouped into batches, each containing at least one CRT. For each batch, (1) the XEVM module executes all transactions on L2, producing new state digests, and (2) the 2PC module drives the 2PC protocol between the two VSMs. We run this workflow for $100$ instances of 2PC and compute the average L1 gas usage of $\pname$. Using the MPT verifier contract, the average gas usage of $\pname$ over all 2PC instances is $1.35$M gas, with approximately $840$K gas for pre-commit and $510$K gas for commit. Thus the MPT-based $\pname$ incurs a $\mathbf{0.64-1.5}\times$ increase in gas usage compared to the referenced gas usage range of Zksync Era, which we illustrate in \cref{fig:2pc-gas-usage-increase}.

\begin{figure}
    \centering
    \includegraphics[scale=0.4]{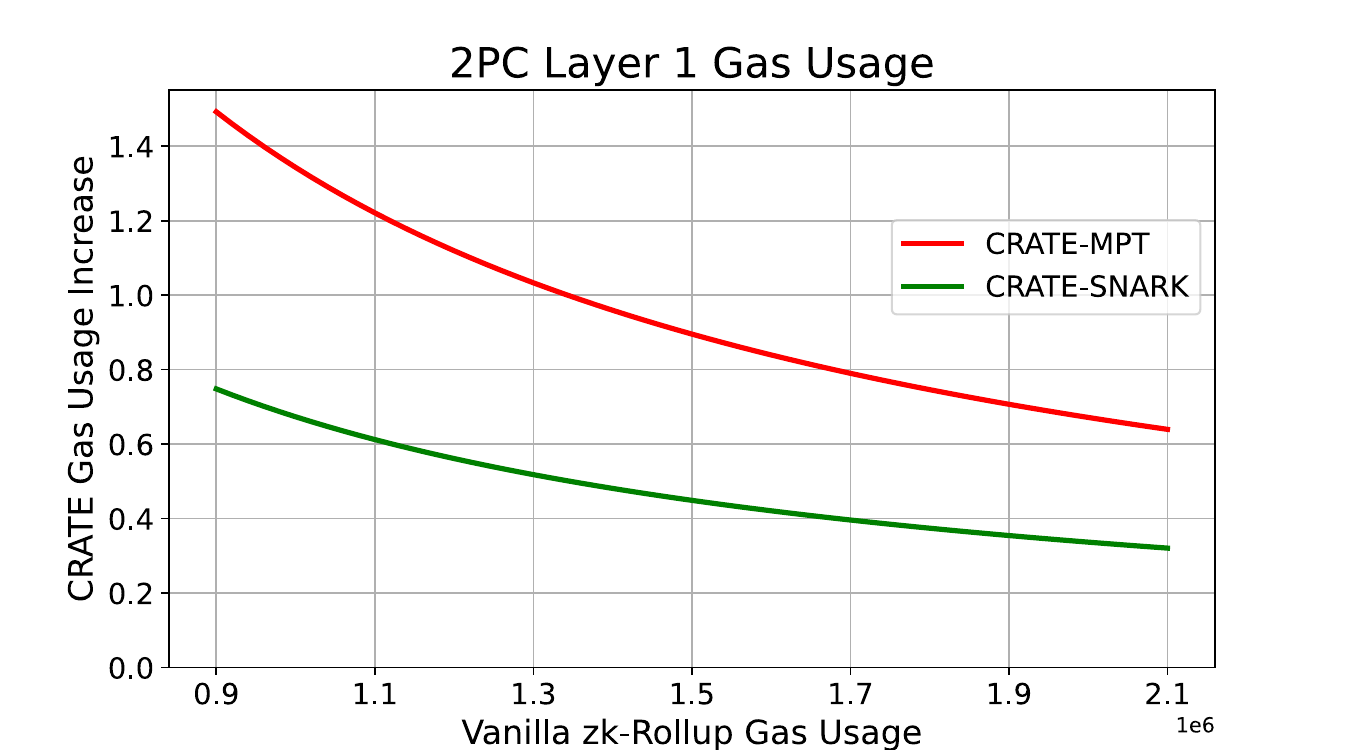}
    \caption{L1 gas usage increase incurred by the MPT-based (red) and SNARK-based (green) implementation of the $\pname$ 2PC protocol, compared to the observed gas usage range of a vanilla zk-rollup (Zksync). }
    \label{fig:2pc-gas-usage-increase}
\end{figure}

\parhead{Using SNARKs instead of native MPT proofs} To further reduce L1 gas usage during 2PC, we replace Merkle membership proofs with SNARK proofs. This allows us to batch the unrelated $\locevd$ and $\pcevd$ membership proofs of the pre-commit round into a single SNARK. Under this implementation, the average gas usage of $\pname$ over all 2PC instances is $670$K gas, with approximately $340$K gas for pre-commit and $330$K gas for commit. Thus the SNARK-based $\pname$ incurs a $\mathbf{0.32-0.75}\times$ increase in gas usage compared to the observed gas usage range of Zksync Era, which we also illustrate in \cref{fig:2pc-gas-usage-increase}.
We note that in practice, a zk-rollup integrating $\pname$ can further reduce gas usage by combining the existing state transition validity SNARK with $\pname$'s pre-commit SNARK into a single SNARK.

Since our experiments above are conducted on local Foundry chains, the depth of our MPT tries is $d = 4$, which is smaller than the depth observed on Ethereum Mainnet account MPT tries ($d=11$ at the time of writing). Larger MPT depths imply larger SNARK circuits, which in turn increase proof generation time. We use the term ``single'' MPT proof to mean a proof for a single MPT trie. To prove membership of a smart contract variable against the chain's state root, we need a ``full'' MPT proof consisting of \emph{two single} MPT proofs; one for the account trie and one for the storage trie. Under this terminology, the SNARK circuit of the pre-commit round contains \emph{two full} MPT proofs (one for $\locevd$ and one for $\pcevd$); the SNARK circuit of the commit round contains \emph{one full} MPT proof (for $\cevd$).

We separately measure the SNARK generation time of a \emph{single} MPT proof for different values of $d \in [2,16]$. For $d=16$, the proof generation time is $130$ seconds (with CPU) and $60$ seconds (with GPU). We use an AWS EC2 g4dn.xlarge instance for both CPU and GPU performance benchmarks. The instance has 4 Intel Xeon Platinum 8259CL CPUs @ 3.095 GHz, 16 GB memory, and 1 NVIDIA Tesla T4 GPU core. Since $\pname$ uses $4$ single MPT proofs for the pre-commit round, we estimate that the SNARK proof generation time of the pre-commit round (for $d=16$) is $4 \times 130 = 520$ seconds (i.e., $8.7$ minutes) with CPU, and $4 \times 60 = 240$ seconds (i.e., $4$ minutes) with GPU. Similarly, the SNARK proof generation time of the commit round is estimated to be $4.4$ minutes with CPU and $2$ minutes with GPU. Note that the \emph{proof size, verification time, and L1 gas usage remain the same} regardless of the MPT depth.

\parhead{Flash Loan application}
We measure the total gas usage of a flash loan CRT across both rollups to be $550$K gas. For reference, single-chain flash loans are reported~\cite{equalizer-fl-cost} to consume $200$K gas on average. This increase is due to the overhead of the six trigger-action calls made by our cross-rollup flash loan. Regardless, our flash loan is reasonably practical, especially considering the significantly lower L2 gas costs.

\section{Discussion}~\label{sec:discussion}
In this section, we discuss extensions to the $\pname$ protocol and implications of our design decisions.

\parhead{Full Serializability}
\emph{Full} serializability requires that no local action can interleave between two CRT actions within a rollup batch, in addition to the relative ordering already required by \emph{weak} serializability. We can easily modify the $\gsc$ contract to support full serializability as follows. If $\action$ is ready to declare a session as inactive, then, before returning, it triggers a new $\textproc{finishSession}$ on the other $\gsc$. This ensures that both GSCs mark the current session as inactive, which was not the case in \cref{alg:gsc-single-trigger}. Further, we require that \emph{all} rollup transactions (including local ones) use $\startsession$ as their entrypoint. Then $\startsession$ can revert any transaction that attempts to execute during an ongoing CRT session.

Unlike L1 protocols with fine-grained state locks~\cite{fal-tccsci-23}, our approach locks the entire rollup state during the CRT session, simplifying implementation. Transactions in the batch can still access the full state outside the CRT session. However, this requires routing all transactions through the $\gsc$, making it not backwards-compatible with existing applications.

\parhead{Distinct Executors}
For ease of exposition, we have described $\pname$ under a shared Executor operating both rollups. However, $\pname$ can easily be adapted to support \emph{distinct}, possibly \emph{distrusting} Executors $E_1$ and $E_2$. Upon listening to a trigger call emitted by $\gsc_1$, $E_1$ sends the triggered action information to $E_2$, who proceeds to execute a matching $\action$. $E_1$ can set proper timeouts so that he does not wait an indefinite amount of time for $E_2$ (and vice-versa). Since $E_1$ and $E_2$ need to communicate back and forth to execute CRTs, the network latency becomes the bottleneck during off-chain execution. Nevertheless, this drawback is not specific to $\pname$ but inherent to the distinct Executors setup.

\parhead{Implications on rollup finality}
Recall from \cref{sec:liveness-chain-crt} that the end-to-end latency of $\pname$ is $4$ rounds. In Ethereum, one round lasts $13$ minutes, so rollups using $\pname$ achieve finality after $4 \times 13 = 52$ minutes; this might seem much longer than the single round ($13$ minutes) required by existing zk-rollups for finality. However, most zk-rollups currently submit new state roots to L1 at intervals in the order of \emph{hours}, since SNARK proof generation is the bottleneck. Thus, both the $4\times$ latency and the $\vsm{}$ state locking introduced by $\pname$ do not affect existing rollups' finality and throughput, since the $\vsm{}$s were already not utilized during that time for any new state updates. Similarly, the increased proving time incurred by our SNARK-based implementation is also dominated by the SNARK validity proof generation time required by existing zk-rollups.

Note that the discussion above is about L1 \emph{finality}, which does not place any trust on the Executor. Today, applications may also be interested in instant confirmation of rollup transactions~\cite{zksync-instant-confirmations}. Since the Executor needs to be trusted for such faster confirmations, $\pname$ preserves the underlying zk-rollups' instant confirmation guarantees.

\section{Related Work}~\label{sec:relwork}
In this section, we compare $\pname$ to existing cross-chain systems. The state-of-the-art cross-rollup composability solution is SVS~\cite{shared-val-seq-23} and variants, such as the CIRC protocol~\cite{espresso-circ-24} by Espresso. We provide a comprehensive comparison of $\pname$ with SVS in~\cref{sec:overview}.

\myparagraph{Cross-chain Interoperability (IO)} Existing cross-chain IO~\cite{chainlink-ccip} protocols primarily adopt the framework of atomic swaps and cross-chain bridges. Atomic swap protocols usually make use of smart contracts, hashlocks, and timelocks~\cite{bitcoin-wiki-htlc,wadhwaHeHTLC2023,tsabaryMADHTLC2021} to ensure the atomic execution of cross-chain transactions. Cross chain bridges~\cite{xiew-zkbridge-22, axelar, near-bridge} connect two blockchains and facilitate the transfer of assets between them. While they enable services such as asset transfers, they cannot offer full cross-chain composability, since they lack a coordination mechanism. 
Current IO protocols~\cite{alt-chains-atomic-transfers, bitcoin-wiki-htlc, herlihy-atomic-cc-swaps-18, herlihy-cc-smr-22, sheff-het-paxos-20, anoma-chimera-chains-23} cannot accommodate complex cross-chain transactions such as flash loans without relying on strong trust assumptions.

\myparagraph{Cross-chain Composability}
Several solutions~\cite{lu-atomic-cc-interactions-24, fal-tccsci-23, zakhary-ac3-20} ensure atomicity for multi-L1 transactions, such as the 2PC4BC protocol~\cite{fal-tccsci-23}, which provides atomicity and serializability but requires $O(n)$ rounds for a transaction of length $n$. These L1 protocols rely on complex locking mechanisms and multiple rounds, making certain applications, like cross-chain flash loans, impractical. In contrast, $\pname$ operates on L2 and requires $O(1)$ rounds by executing transactions off-chain and including them in a single rollup batch. 
For instance, 2PC4BC requires at least 5 L1 rounds ($\sim60$ blocks) for flash loan execution and finality, whereas $\pname$ performs execution off-chain, with L1 involvement only needed for finality.

Shared sequencer networks like Espresso~\cite{espresso-docs}, Astria~\cite{astria-docs}, and Radius~\cite{radius-docs} leverage a shared sequencer to achieve censorship resistance and faster pre-confirmation. While these solutions enable certain types of cross-rollup transactions, their guarantees are limited to atomic inclusion (not execution) and rely on trusting the shared sequencer.

The Superchain~\cite{superchain-explainer} is a network of layer 2 chains built on the OP Stack~\cite{optimism-website}, enabling resource sharing and asynchronous messaging but relying on a shared sequencer for cross-chain atomicity. Polygon AggLayer~\cite{agglayer} connects layer 1 and layer 2 networks via a unified bridge, supporting token transfers and message passing, but its cross-chain transactions require Ethereum settlement, leading to high asynchrony.

\section{Conclusion}~\label{sec:conclusion}
We have introduced $\pname$, a secure protocol for cross-rollup composability, enabling users to atomically execute transactions spanning two rollups, guaranteeing the all-or-nothing and serializability properties. Our protocol is flexible since it supports rollups on distinct L1s, as well as rollups operated by distinct Executors. We provided two formal models for cross-rollup transactions (CRTs), and defined atomicity within them. Our novel formal treatment gave rise to the shortcomings of other works and allowed us to rigorously prove security for $\pname$. We implemented and benchmarked all components of $\pname$ and an end-to-end cross-rollup flash loan application to demonstrate the practicality of $\pname$.

\section*{Acknowledgements}
We want to thank Fangyan Shi and Wenhao Wang for their help with the SNARK-based implementation of MPT verification, as well as Sen Yang for helpful discussions.

\bibliographystyle{plain}
\bibliography{ref,fanz}

\begin{thebibliography}{10}

\bibitem{alt-chains-atomic-transfers}
Alt chains and atomic transfers.
\newblock \url{https://bitcointalk.org/index.php?topic=%20193281.msg2224949#msg2224949}.

\bibitem{axelar}
Axelar {\textbar} {Secure} cross-chain communication for {Web3}.
\newblock \url{https://axelar.network/}.

\bibitem{circom-site}
{Circom 2 Documentation}.
\newblock \url{https://docs.circom.io/}.

\bibitem{forbes-crypto-cap-24}
Cryptocurrency {Prices}, {Market} {Cap} and {Charts}.
\newblock \url{https://www.forbes.com/digital-assets/crypto-prices/}.

\bibitem{defillama-bridges}
{DefiLlama}.
\newblock \url{https://defillama.com/bridges}.

\bibitem{near-bridge}
{ETH} <> {NEAR} {Rainbow} {Bridge}.
\newblock \url{https://pages.near.org/bridge/}.

\bibitem{etherscan}
The ethereum blockchain explorer.
\newblock \url{https://etherscan.io/}.

\bibitem{equalizer-fl-cost}
Flash loan gas cost comparison.
\newblock \url{https://docs.equalizer.finance/equalizer-deep-dive/flash-loan-gas-cost-comparison}.

\bibitem{bitcoin-wiki-htlc}
Hash {Time} {Locked} {Contracts} - {Bitcoin} {Wiki}.
\newblock \url{https://en.bitcoin.it/wiki/Hash_Time_Locked_Contracts}.

\bibitem{optimism-website}
Optimism.
\newblock \url{https://www.optimism.io/}.

\bibitem{zksync-explorer}
Zksync era block explorer.
\newblock \url{https://explorer.zksync.io/}.

\bibitem{anoma-chimera-chains-23}
Typhon's {Chimera} {Chains} {\textbar} {Blog} - {Anoma}, 2023.
\newblock \url{https://anoma.net/blog/chimera-chains}.

\bibitem{appel-continuations-2007}
Andrew~W. Appel.
\newblock {\em Compiling with Continuations}.
\newblock Cambridge University Press, USA, 2007.

\bibitem{astria-docs}
Astria.
\newblock Astria documentation.
\newblock \url{https://docs.astria.org/}.

\bibitem{augusto-sok-io-24}
André Augusto, Rafael Belchior, Miguel Correia, André Vasconcelos, Luyao Zhang, and Thomas Hardjono.
\newblock Sok: Security and privacy of blockchain interoperability.
\newblock In {\em 2024 IEEE Symposium on Security and Privacy (SP)}, pages 3840--3865, 2024.

\bibitem{bernstein-concurrency-databases-1986}
Philip~A Bernstein, Vassos Hadzilacos, and Nathan Goodman.
\newblock {\em Concurrency control and recovery in database systems}.
\newblock Addison-Wesley Longman Publishing Co., Inc., USA, June 1986.

\bibitem{buterin-gasper-2020}
Vitalik Buterin, Diego Hernandez, Thor Kamphefner, Khiem Pham, Zhi Qiao, Danny Ryan, Juhyeok Sin, Ying Wang, and Yan~X. Zhang.
\newblock Combining {GHOST} and casper.
\newblock {\em CoRR}, abs/2003.03052, 2020.

\bibitem{chainlink-ccip}
Chainlink.
\newblock Cross-{Chain} {Interoperability} {Protocol} ({CCIP}) {\textbar} {Chainlink}.
\newblock \url{https://chain.link/cross-chain}.

\bibitem{espresso-circ-24}
Espresso.
\newblock {CIRC}: Coordinated inter-rollup communication.
\newblock \url{https://espresso.discourse.group/t/circ-coordinated-inter-rollup-communication/43}.

\bibitem{fal-tccsci-23}
Ghareeb Falazi, Uwe Breitenb\"{u}cher, Frank Leymann, Stefan Schulte, and Vladimir Yussupov.
\newblock Transactional cross-chain smart contract invocations.
\newblock {\em Distrib. Ledger Technol.}, 8 2023.

\bibitem{eth-org-smart-contracts}
Etherem Foundation.
\newblock Introduction to smart contracts.
\newblock \url{https://ethereum.org/en/smart-contracts}.

\bibitem{eth-evm}
Ethereum Foundation.
\newblock Ethereum virtual machine (evm).
\newblock \url{https://ethereum.org/en/developers/docs/evm/}.

\bibitem{eth-rpc-api}
Ethereum Foundation.
\newblock {JSON-RPC API}.
\newblock \url{https://ethereum.org/en/developers/docs/apis/json-rpc/}.

\bibitem{eth-org-layer-2}
Ethereum Foundation.
\newblock Layer 2.
\newblock \url{https://ethereum.org/en/layer-2/}.

\bibitem{mpt-trie}
Ethereum Foundation.
\newblock {Merkle Patricia Trie}.
\newblock \url{https://ethereum.org/en/developers/docs/data-structures-and-encoding/patricia-merkle-trie/}.

\bibitem{eth-org-reorg}
Ethereum Foundation.
\newblock Reorgs in proof of stake.
\newblock \url{https://ethereum.org/en/developers/docs/consensus-mechanisms/pos/attack-and-defense/#reorgs}.

\bibitem{eth-org-zk-rollups}
Ethereum Foundation.
\newblock Zero-{Knowledge} rollups.
\newblock \url{https://ethereum.org/en/developers/docs/scaling/zk-rollups}.

\bibitem{superchain-explainer}
Optimism Foundation.
\newblock Superchain explainer.
\newblock \url{https://docs.optimism.io/superchain/superchain-explainer}.

\bibitem{ghaemi-pubsub-blockchain-2021}
Sara Ghaemi, Sara Rouhani, Rafael Belchior, Rui~Santos Cruz, Hamzeh Khazaei, and Petr Mus{\'{\i}}lek.
\newblock A pub-sub architecture to promote blockchain interoperability.
\newblock {\em CoRR}, abs/2101.12331, 2021.

\bibitem{groth-16}
Jens Groth.
\newblock On the size of pairing-based non-interactive arguments.
\newblock In Marc Fischlin and Jean-S{\'e}bastien Coron, editors, {\em Advances in Cryptology -- EUROCRYPT 2016}, pages 305--326, Berlin, Heidelberg, 2016. Springer Berlin Heidelberg.

\bibitem{herlihy-atomic-cc-swaps-18}
Maurice Herlihy.
\newblock Atomic {Cross}-{Chain} {Swaps}.
\newblock In {\em Proceedings of the 2018 {ACM} {Symposium} on {Principles} of {Distributed} {Computing}}, {PODC} '18, pages 245--254, New York, NY, USA, July 2018. Association for Computing Machinery.

\bibitem{jansen-ibm-event-driven-2020}
G.~Jansen and J.~Saladas.
\newblock Advantages of event-driven architecture.
\newblock Developer IBM, 2020.
\newblock \url{https://developer.ibm.com/technologies/messaging/articles/advantages-of-an-event-driven-architecture/}.

\bibitem{solidity-mpt}
Jun Kimura.
\newblock {solidity-mpt}.
\newblock \url{https://github.com/ibc-solidity/solidity-mpt}.

\bibitem{zksync-era}
Matter Labs.
\newblock Zksync.
\newblock \url{https://zksync.io/ecosystem}.

\bibitem{zksync-instant-confirmations}
Matter Labs.
\newblock Instant confirmations, 2024.
\newblock \url{https://docs.zksync.io/zksync-protocol/rollup/finality#instant-confirmations}.

\bibitem{quarkslab-zksync-workflow}
Madigan Lebreton, Elouan Wauquier, and Victor Houal.
\newblock zksync transaction workflow.
\newblock \url{https://blog.quarkslab.com/zksync-transaction-workflow.html}.
\newblock Published: October 26, 2023.

\bibitem{lu-atomic-cc-interactions-24}
Huaixi Lu, Akshay Jajoo, and Kedar~S. Namjoshi.
\newblock Atomicity and abstraction for cross-blockchain interactions, 2024.

\bibitem{merkle-tree}
Ralph~C. Merkle.
\newblock A digital signature based on a conventional encryption function.
\newblock In Carl Pomerance, editor, {\em Advances in Cryptology --- CRYPTO '87}, pages 369--378, Berlin, Heidelberg, 1988. Springer Berlin Heidelberg.

\bibitem{agglayer}
Polygon.
\newblock Aggregation layer.
\newblock \url{https://polygon.technology/agglayer}.

\bibitem{qin-attacking-defi-flash-loans-21}
Kaihua Qin, Liyi Zhou, Benjamin Livshits, and Arthur Gervais.
\newblock Attacking the {DeFi} {Ecosystem} with {Flash} {Loans} for {Fun} and {Profit}.
\newblock In Nikita Borisov and Claudia Diaz, editors, {\em Financial {Cryptography} and {Data} {Security}}, Lecture {Notes} in {Computer} {Science}, pages 3--32, Berlin, Heidelberg, 2021. Springer.

\bibitem{radius-docs}
The Radius.
\newblock The radius documentation.
\newblock \url{https://docs.theradius.xyz/}.

\bibitem{rah-even-driven-microservices-2022}
Alam Rahmatulloh, Fuji Nugraha, Rohmat Gunawan, and Irfan Darmawan.
\newblock Event-driven architecture to improve performance and scalability in microservices-based systems.
\newblock In {\em 2022 International Conference Advancement in Data Science, E-learning and Information Systems (ICADEIS)}, pages 01--06, 2022.

\bibitem{sheff-het-paxos-20}
Isaac Sheff, Xinwen Wang, Robbert van Renesse, and Andrew~C. Myers.
\newblock Heterogeneous {Paxos}: {Technical} {Report}, November 2020.

\bibitem{shiFoundations2020}
Elaine Shi.
\newblock {\em Foundations of {Distributed} {Consensus} and {Blockchains}}.
\newblock 2020.

\bibitem{sussman-scheme-1998}
Gerald~Jay Sussman and Guy~L. Steele.
\newblock Scheme: {A} {Interpreter} for {Extended} {Lambda} {Calculus}.
\newblock {\em Higher-Order and Symbolic Computation}, 11(4):405--439, December 1998.

\bibitem{espresso-docs}
Espresso Systems.
\newblock Espresso network documentation.
\newblock \url{https://docs.espressosys.com/network}.

\bibitem{foundry-book}
Foundry Team.
\newblock Foundry book.
\newblock \url{https://book.getfoundry.sh/}.

\bibitem{thaler-proofs-args-zk-22}
Justin Thaler.
\newblock {\em Proofs, Arguments, and Zero-Knowledge}.
\newblock Georgetown University, 2022.

\bibitem{tsabaryMADHTLC2021}
Itay Tsabary, Matan Yechieli, Alex Manuskin, and Ittay Eyal.
\newblock {MAD}-{HTLC}: {Because} {HTLC} is {Crazy}-{Cheap} to {Attack}.
\newblock In {\em 2021 {IEEE} {Symposium} on {Security} and {Privacy} ({SP})}, pages 1230--1248, May 2021.
\newblock ISSN: 2375-1207.

\bibitem{shared-val-seq-23}
John~Guibas Uma~Roy, 0xShitTrader.
\newblock Shared validity sequencing {\textbar} umbra research.
\newblock \url{https://umbraresearch.xyz/writings/shared-validity-sequencing}.

\bibitem{wadhwaHeHTLC2023}
Sarisht Wadhwa, Jannis Stoeter, Fan Zhang, and Kartik Nayak.
\newblock He-{HTLC}: {Revisiting} {Incentives} in {HTLC}.
\newblock In {\em Network and {Distributed} {System} {Security} ({NDSS}) {Symposium} 2023}, San Diego, CA, USA, 2023.

\bibitem{wood-ethereum-2014}
Dr~Gavin Wood.
\newblock Ethereum: {A} secure decentralised generalised transaction ledger.
\newblock {\em Ethereum project yellow paper}, 151:1--32, 2014.

\bibitem{a16z-composability}
Linda Xie.
\newblock Composability is {Innovation}, June 2021.
\newblock \url{https://a16zcrypto.com/posts/article/how-composability-unlocks-crypto-and-everything-else/}.

\bibitem{xiew-zkbridge-22}
Tiancheng Xie, Jiaheng Zhang, Zerui Cheng, Fan Zhang, Yupeng Zhang, Yongzheng Jia, Dan Boneh, and Dawn Song.
\newblock {zkBridge}: {Trustless} {Cross}-chain {Bridges} {Made} {Practical}.
\newblock In {\em Proceedings of the 2022 {ACM} {SIGSAC} {Conference} on {Computer} and {Communications} {Security}}, {CCS} '22, pages 3003--3017, New York, NY, USA, November 2022. Association for Computing Machinery.

\bibitem{herlihy-cc-smr-22}
Yingjie Xue and Maurice Herlihy.
\newblock Cross-chain state machine replication, 2022.
\newblock \url{https://arxiv.org/abs/2206.07042}.

\bibitem{zakhary-ac3-20}
Victor Zakhary, Divyakant Agrawal, and Amr El~Abbadi.
\newblock Atomic commitment across blockchains.
\newblock {\em Proceedings of the VLDB Endowment}, 13(9):1319--1331, May 2020.

\end{thebibliography}

\appendix
\newpage

\section{General System Contract}~\label{apdx:gsc}
 \cref{alg:gsc-mult-triggers} contains the pseudocode for the General System Contract supporting DAG-CRTs. Modifications on top of the chain-CRT version are highlighted with \textcolor{dag}{brown color}.

\begin{algorithm}
    \small
    \caption{General System Contract: DAG-CRT}
    \label{alg:gsc-mult-triggers}
    \begin{algorithmic}[1]
    \Statex
    \Function{\textcolor{chain}{startSession}}{\textcolor{chain}{$\addr{}, \cdata$}}
            \State \textcolor{chain}{ \require ($\Not \; \acalled{}$) }
            \State \textcolor{chain}{ $\enonce{} \gets \enonce{} + 1$ }
            \State \textcolor{chain}{ $\snonce{} \gets \snonce{} + 1$ }
            \State \textcolor{chain}{ $\sactive{} \gets \true$ }
            \State \textcolor{chain}{ $\tcalled{} \gets \false$ }
            \State \textcolor{dag}{ (status, \_) $\gets \addr{}.\call(\cdata)$ }
            \State \textcolor{dag}{ \require (status) }
            \State \textcolor{dag}{ \require ($\tcalled{}$) }
    \EndFunction
    \Statex
    \Function{trigger}{$\addr{}, \cdata$}
        \If{\textcolor{dag}{ ($\Not \; \tcalled{}$) } }
            \State \textcolor{dag}{$\tnonce{} \gets \tnonce{} + 1$} {\Comment{Allow many triggers per tx}}
        \EndIf
        \State \textcolor{chain}{$\tcalled{} \gets \true$}
        \State \textcolor{chain}{$\sid{} \gets \snonce{}$}
        \State $ \msghash \gets \hash(\msgsender, \addr{}, \cdata, \tnonce{}, \textcolor{chain}{\sid{}})$
        \State \textbf{emit T}($\msgsender, \addr{}, \cdata, \tnonce{}$)
        \State $\ttree.\insertt(\msghash)$
    \EndFunction
    \Statex
    
    \Function{action}{\textcolor{dag}{$\tdatalist{}$}, \textcolor{chain}{$\sid{}$}}
            \State \textcolor{chain}{\Call{checkSessionId}{$\sid{}$}}
            \State \textcolor{chain}{$\acalled{} \gets \true$}
            \State \textcolor{chain}{$\tcalled{} \gets \false$}
            \State \textcolor{dag}{ \require ($\mathsf{len}(\tdatalist{}) > 0$) }
            \State \textcolor{dag}{ $\anonce{} \gets \anonce{} + 1$ }
            \For{\textcolor{dag}{ trigData : trigDataList } }
                \State \textcolor{dag}{(sender, caddr, cdata) $\gets$ trigData }
                \State $\xsender{} \gets \sender$
                \State (status, \_) $\gets \addr{}.\call(\cdata)$
                \State \require (status)
                \State $ \msghash \gets \hash(\sender, \addr{}, \cdata, \anonce{}, \textcolor{chain}{\sid{}}) $
                \State $\ttree.\insertt(\msghash)$
                \State $\xsender{} \gets \texttt{address}(0)$
            \EndFor
            \If{\textcolor{chain}{ ($\Not \; \tcalled{}$) } }
                \State \textcolor{chain}{ $\sactive{} \gets \false$ }
            \EndIf
            \State \textcolor{chain}{$\acalled{} \gets \false$}
    \EndFunction
    \end{algorithmic}
\end{algorithm}

\section{Validator Smart Contract}~\label{apdx:vsm}
\cref{alg:vsm-trustless-p1} fully specifies the functions used in \cref{alg:vsm-generic}. Here, we assume the existence of two functions for abstraction purposes: $\textproc{VerAttributesWithBridge}(\cdot)$, which verifies state attributes against bridge-relayed block headers, and $\textproc{VerAttributes}(\cdot)$, which verifies the MPT membership of state attributes against rollup state roots (no bridge involved).  

\algblock{Variables}{EndVariables}

\begin{algorithm}[H]
    \small
    \caption{Validator Smart Contract: Concrete 2PC}
    \label{alg:vsm-trustless-p1}
    \begin{algorithmic}[1] %

        \Variables
            \State \{$\vsmid, \dst{}, \idx, \dstemp{} $\}; $\pstat \in \{\free, \paired \};$
            \State \{$\troot{}', \aroot{}', \snonce{}', \enonce{}'$\};
            \State $\decisions$: mapping $\idx \Rightarrow \commit/\abort$
        \EndVariables
        
        \Statex
        \Function{IsLocal}{$\dstt{}, \locevd$}
            \State $(\roots, \nonces, \attproofs) \gets \locevd$
            \State $\textproc{VerAttributes}(\dst{}, \dstt{}, \roots, \nonces, \attproofs)$
            \State $\troot{}, \aroot{},  \troot{}', \aroot{}' \gets \roots$
            \State $\snonce{}', \enonce{}' \gets \nonces$
            \State Store $(\troot{}', \aroot{}', \snonce{}', \enonce{}')$
            \State \Return $(\troot{} = \troot{}') \And (\aroot{} = \aroot{}')$
        \EndFunction
        
        \Statex
        \Function{VerPreComEvd}{$\dstt{}, \pcevd$}
            \State $(\ovsm, \attproofs, \dstt{O}) \gets \pcevd$
            \State \Call{VerAttributesWithBridge}{$\ovsm, \attproofs$}
            \State $\idx \gets \Call{GetIndex}{\dst{}, \dstt{}', \ovsm.\dst{}, \dstt{O}}$
            \State $\leader \gets \Call{GetLeader}{\idx}$
            \If{$\leader = \vsmid$} \; \require $\ovsm.\pstat = \free $
            \Else
                \State \require $\ovsm.\pstat = \paired $ 
                \State \require $\ovsm.\idx = \idx$
            \EndIf
            
        \EndFunction    

        \Statex

        \Function{VerComEvd}{$\cevd$}
            \State $(\ovsm, \attproofs) \gets \cevd$
            \State \Call{VerAttributesWithBridge}{$\ovsm, \attproofs$}
            \State $\leader \gets \Call{GetLeader}{\idx}$
            \If{$\leader = \vsmid$}
                \State \require $\ovsm.\pstat = \paired$ 
                \State \require $\ovsm.\idx = \idx$
                \State \require $\troot{}' = \ovsm.\aroot{}'$
                \State \require $\aroot{}' = \ovsm.\troot{}'$
                \State \require $\enonce{}' + \ovsm.\enonce{}' = \snonce{}'$
            \Else \; \require $\ovsm.\decisions[\idx] = \commit$
            \EndIf
            \State \Return $\true$
        \EndFunction
        \Statex
        \Function{VerAbEvd}{$\dstt{}, \abevd$}
            \State $(\ovsm, \attproofs) \gets \abevd$
            \State \Call{VerAttributesWithBridge}{$\ovsm, \attproofs$}
            \State $\leader \gets \Call{GetLeader}{\idx}$
            \If{$\leader = \vsmid$}
                \If{$\big( (\ovsm.\pstat = \paired)$ $\And$  $(\ovsm.\idx = \idx) \big)$}
                    \State \require {$(\troot{}', \aroot{}') \neq (\ovsm.\aroot{}', \ovsm.\troot{}')$}
                \EndIf
                \State $\decisions[\idx] \gets \abort$
            \Else
                \State \require ($\ovsm.\decisions[\idx] = \abort$)
            \EndIf
        \EndFunction

    \end{algorithmic}
\end{algorithm}

\section{XEVM: Executor Algorithm}
\begin{algorithm}[H]
    \small
    \caption{Shared Executor: XEVM -- Chain CRTs}
    \label{alg:shared-executor-xevm-chain-crt}

    \begin{algorithmic}[1]
    \Variables
        \State $\rolids := \{1,2\}$: rollup ids
        \State $\state{1}, \state{2}$: states of rollups $\rol{1}, \rol{2}$
        \State $\cp{1}, \cp{2}$: state checkpoints
        \State $\gsc_1, \gsc_2$: General System contracts
    \EndVariables
    \Statex
    \Function{xevm}{$\tx{}$}
        \State $(\status{1}, \outt{}) \gets $ \Call{evm}{$\tx{}$}
        \State $\tx' \gets $ parse($\outt{}$)
        \If{$\tx' = \nul$} \Return $\status{1}$ \EndIf
        \State \textbf{Assert} $(\tx'.\rol{} \in \rolids) \And (\tx'.\rol{} \neq \tx.\rol{})$
        \State $\tx'' \gets \gsc_{\tx'.\rol{}}.\action(\langle \tx' \rangle)$
        \State $\status{2} \gets$ \Call{xevm}{$\tx''$}
        \State \Return $\status{2}$
    \EndFunction
    \Statex
    \Function{entryPoint}{$\tx{}$}
        \For{$\id : \rolids$}
            \State $\cp{\id} \gets \state{\id}$
        \EndFor
        \State $\status{} \gets $ \Call{xevm}{$\tx{}$}
        \If{$(\status{} = \false $}
            \For{$\id : \rolids$}
                \State $\state{\id} \gets \cp{\id}$
            \EndFor
        \EndIf
        
    \EndFunction
\end{algorithmic}
    
\end{algorithm}

\newcommand{\txlist}{\mathsf{txList}}

\begin{algorithm}[ht]
    \small
    \caption{DAG-\pname \; Shared Executor: XEVM}
    \label{alg:shared-executor-xevm-dag}

    \begin{algorithmic}[1]
    \Function{xevm}{$\tx{}$}
        \State $(\status{1}, \outt{}) \gets $ \Call{evm}{$\tx{}$}
        \State \textcolor{dag}{ $\txlist \gets $ parse($\outt{}$) }
        \If{\textcolor{dag}{ $|\txlist| = 0$} } \Return $\status{1}$ \EndIf
        \State \textcolor{dag}{ $\rol{}' \gets \txlist'.\rol{}$ }
        \State \textbf{Assert} $(\rol{}' \in \rolids) \And (\rol{}' \neq \tx.\rol{}) $
        \State $\tx'' \gets \gsc_{\rol{}'}.\action(\textcolor{dag}{[\langle \tx' \rangle]_{\tx' \in \txlist} })$
        \State $\status{2} \gets$ \Call{xevm}{$\tx''$}
        \State \Return $\status{2}$
    \EndFunction
\end{algorithmic}
    
\end{algorithm}

\section{Chain-CRT Deferred Proofs}~\label{apdx:chain-crt-proofs}

For an executed action $a$, we say $a$ is \emph{action-wrapped} if $\act{}(a.\desc{})$ is in fact executed -- instead of plain $a$. 
Similarly, we say $a$ is \emph{triggering} if $a.\descnext{} \neq \nul$, i.e., it makes at least one call to $\gsc.\trig(\cdot)$. Given a CRT $\ccrt = [a_i]$, notice that $a_1$ is the only CRT action that is \emph{not} action-wrapped. This is because $a_1$ is the user-signed entrypoint action and is executed as is, thus its $\desc{}$ field is \emph{not} inserted in the action tree. Similarly, $a_n$ is the only non-triggering action of the CRT.

\begin{proof}[Proof of \cref{lem:svs-gsc-honest-or-reverse-order}]
    The ``only if'' direction corresponds to the honest execution; its proof is simple and omitted. We focus on the ``if'' direction. Without loss of generality, assume $n = 2k$ is even and $a_1$ is executed on $\rol{1}$. This means that the last action $a_{2k}$ is executed on $\rol{2}$. 
    Let $[i_1, \dots, i_k]$ denote the relative order of actions $\{a_{2i-1}\}_{i \in [k]}$ in $\batch{1}$, and similarly let $[j_1, \dots, j_k]$ be the relative order of actions $\{a_{2j}\}_{j \in [k]}$ in $\batch{2}$. Let $f \in [k]$ (for ``first'') be the position of the first action, i.e, $i_f = 1$, and let $l \in [k]$ (for ``last'') denote the position of the last action, i.e., $j_l = 2k$.

    Since the last action is on $\rol{2}$, all $\rol{1}$ actions $[a_{i_r}]_{r \in [k]}$ trigger another action, so the (non-null) descriptions $[a_{i_r}.\descnext{}]_{r \in [k]}$ are inserted into $\ttree_1$ in this order. Similarly, since the first action is on $\rol{1}$, all $\rol{2}$ actions are action-wrapped, so the the (non-null) descriptions $[a_{j_r}.\desc{}]_{r \in [k]}$ are inserted into $\atree_2$ in this order. Thus we must have $j_r = i_r + 1$ for all $r \in [k]$, since otherwise $\ttree_1' \neq \atree_2'$. 
    
    \noindent\textbf{Case 1:} $l = 1$. Then each $\rol{2}$ action except for $a_{j_1} = a_{2k}$ triggers another action, so the (non-null) descriptions $[a_{j_r}.\descnext{}]_{2 \leq r  \leq k}$ are inserted in $\ttree_2$ in that order. Similarly, each $\rol{1}$ action except for $a_{i_f}$ action-wrapped, so the (non-null) descriptions $[a_{i_r}.\desc{}]_{r \neq f}$ are inserted in $\atree_1$ in that order. In order to have $\ttree_2' = \atree_1'$, it must be $i_r = j_{r+1} + 1$ for all $r < f$ and $i_r = j_r$ for all $r > f$. If $f < k$, then we have both $i_{k} = j_{k}$ and $j_{k} = i_{k} + 1$, which is impossible; thus it must be $f = k$. Combining we get $i_{r+1} = i_r - 2$ and $j_{r+1} = j_r - 2$ which together with $l=1$ lead to closed forms $i_r = 2(j-r) + 1$ and $j_r = 2(j-r) + 2$. This means $[i_r]_{r \in [k]}$ (resp. $[j_r]_{r \in [k]}$) is exactly the \emph{reverse} of the honest order $[2i - 1]_{i \in [k]}$ (resp. $[2i]_{i \in [k]}$).
    
    \noindent\textbf{Case 2:} $l \geq 2$. Similar to case 1, the (non-null) descriptions $[a_{j_r}.\descnext{}]_{r \neq l}$ are inserted in $\ttree_2$ in that order, and the (non-null) descriptions $[a_{i_r}.\desc{}]_{r \neq f}$ are inserted in $\atree_1$ in that order. If $f \geq 2$, and since $l \geq 2$, we have both $i_{1} = j_{1}$ and $j_{1} = i_{1} + 1$, which is impossible; thus it must be $f = 1$, i.e., $a_{i_1} := a_1$. Similarly, if $l < k$, we reach the impossibility of $i_{k} = j_{k}$ and $j_{k} = i_{k} + 1$; thus it must be $l = k$. Combining yields closed forms $i_r = 2r - 1$ and $j_r = 2r$, i.e., both $[i_r]_{r \in [k]}$ and $[j_r]_{r \in [k]}$ precisely match the honest order.
\end{proof}

\begin{proof}[Proof of \cref{lem:chain-gsc-honest-order}]
    The ``only if'' direction corresponds to the honest execution; its proof is simple and thus omitted. We prove the contrapositive of the ``if'' direction, i.e., we begin by assuming $\batch{1}$ and $\batch{2}$ \emph{do not} preserve the relative order of actions, and we wish to show that at least one of (I,II,III) is not preserved. 
    Similar to the proof of \cref{lem:svs-gsc-honest-or-reverse-order}, assume $n = 2k$ is even and $a_1$ is executed on $\rol{1}$. This means that the last action $a_{2k}$ is executed on $\rol{2}$. If the actions' order is \emph{not} fully reversed in the two batches, then \cref{lem:svs-gsc-honest-or-reverse-order} implies that (I) is not preserved, so we are done. For the rest of the proof, we assume the actions' order is fully reversed in the two batches. Given that (III) holds under $(\dst{1}, \dst{2})$, we have two cases:
    
    \noindent\textbf{Case 1:} $\sactive{1} = \false$. The first CRT action from $\batch{1}$ to be executed is $a_{2k-1}$, which is action-wrapped. Notice that it must use $\sid{1} = \snonce{1} + 1$; if it uses $\sid{1} = \snonce{1}$, then $\textproc{checkSessionID}$ will fail since $\sactive{1} = \false$. Next, $\textproc{checkSessionID}$ will set $\snonce{1}' = \snonce{1} + 1$ and $\sactive{1}' = \true$. Since $\sid{1}$ is part of the tuple inserted in $\atree_1$, all other action-wrapped actions must be using $\sid{1}$ as their session id, too; if not, then the trees will not match and (I) is not preserved. After executing all actions except $a_1$, we have that $\snonce{1}' = \snonce{1} + 1$ and $\snonce{2}' = \snonce{2} + 1$. We now consider the execution of $a_1$ on $\rol{1}$, which is the last CRT action in $\batch{1}$.
    \noindent\textbf{Subcase 1(a):} The user action $a_1$ is \emph{proper} and calls $\textproc{startSession}$ to trigger $a_2$. Then $\enonce{1}' = \enonce{1} + 1$ and \emph{further} the session nonce now is $\snonce{1}' = \snonce{1} + 2$. This new session nonce is inserted in $\ttree_1$, which causes a mismatch with $\atree_2$, where $a_2$ was inserted along with $\snonce{2}' = \snonce{1} + 1$. Thus (I) is not preserved.
    \noindent\textbf{Subcase 1(b):} The user action $a_1$ is \emph{improper} and directly calls $\textproc{trigger}$ to trigger $a_2$. Then $\enonce{1}' = \enonce{1}$ is unchanged and still $\snonce{}' = \snonce{} + 1$. Thus we have $\enonce{1}' + \enonce{2}' \neq \snonce{1}' $ and (II) is not preserved.
    
    \noindent\textbf{Case 2:} $\sactive{1} = \true$. Then $\sactive{2} = \false$ and now $a_{2k}$ must use $\sid{2} = \snonce{2} + 1$. The argument from here on is similar to case 1.

    In either case, either property (I) or (II) is not preserved.
\end{proof}

\section{DAG-CRT Deferred Proofs}~\label{apdx:dag-crt-proofs}
\noindent\textbf{Notation.} We use the following terminology and notation:
\begin{itemize}[leftmargin=*]
    \item For an action $a$ and sub-action $sa$, if $sa.\desc{} \in a.\desc{}$, we say that $a$ is the \emph{parent} action of $sa$.   
    \item We use $sa.\rol{}$ to denote the rollup on which a sub-action $sa$ is to be executed, which is the same as the rollup field of the parent action.
    \item We use $sa \in a$ to mean $sa.\desc{} \in a.\desc{}$, as well as $a[j] = sa$ to mean $a.\desc{j} = sa.\desc{}$.
    \item We write $a.\desc{} \subseteq \atree$ to mean that $a.\desc{} = [sa_j.\desc{}]_{j \in [|a.\desc{}|]}$ appears as a contiguous subsequence in the array maintained by $\atree$. We use similar notation for $\descnext{}$ and $\ttree$.  
\end{itemize} 

We also use the notation ``$a.\desc{} \subseteq^* \atree$'' to mean that $a.\desc{} \subseteq \atree$ and all $a.\desc{}[j]$ -- and only them --  are inserted with the same $\anonce{}$. The same goes for ``$a.\descnext{} \subseteq^* \ttree$''.

\cref{claim:a-exec-iff-desc-inserted-dag} conveys the same core idea as \cref{claim:a-exec-iff-desc-inserted-chain}, but the argument becomes more intricate due to the complexity of DAG-CRTs.

\begin{claim}~\label{claim:a-exec-iff-desc-inserted-dag}
    Let $\ttree$ and $\atree$ be the trigger and action trees of $\gsc$ after the transaction batch $\batch{}$ is accepted by the $\vsm{}$ of a rollup, and let $a$ be a \emph{non-local} action. Then $a \in \batch{}$ if and only if $a.\desc{} \subseteq^* \atree$ and $a.\descnext{} \subseteq^* \ttree$.~\footnote{This statement also captures the edge cases where either $a$ is not triggering or is not action-wrapped; e.g. if $|a.\desc{}| = 0$, then trivially $a.\desc{} = [] \subseteq \atree$.}
\end{claim}

\begin{proof}[Proof of \cref{claim:a-exec-iff-desc-inserted-dag}]
     The ``only if'' direction is simple and omitted. We show the contrapositive of the ``if'' direction, i.e., we begin by assuming $a \notin \batch{}$; also assume $a.\desc{} \subseteq \atree$ (if not, we are done). Since $a \notin \batch{}$ \emph{as is}, there are two cases under which the subactions of $a$ still appear as a contiguous subsequence under $\atree$. (1) The subactions of $a$ are partitioned and each partition is part of a distinct parent action in the batch; in this case, by construction of \cref{alg:gsc-mult-triggers}, there exist two subactions belonging to different parent actions, which are inserted in $\atree$ with different $\anonce{}$ values. Thus $a.\desc{} \not \subseteq^* \atree$. (2) All subactions of $a$ are part of the same parent action in the batch (and thus share the same $\anonce{}$), but that parent action additionally includes at least one other subaction (preceding or following the subactions of $a$); in this case, the subactions of $a$ are not the \emph{only} ones to use this $\anonce{}$ value, and we again have $a.\desc{} \not \subseteq^* \atree$.
\end{proof}

\begin{lemma}~\label{lem:dag-gsc-all-or-none-exec}
    Let $\dagcrt = [a_i]_{i \in [n]}$ and $\{\rol{b}, \dst{b}, \dstt{b}, \batch{b}\}_{b \in \{1,2\}}$ be as in \cref{def:forward-atom-chain-crt}, when using the DAG-CRT $\gsc$ contract for off-chain execution. For $b \in \{1,2\}$, let $\{\ttree_b, \atree_b \}$ be the $\gsc_b$ trees under $\dst{b}$.
    Also assume 
    
    \noindent(I) $\ttree_1 = \atree_2$ and $\ttree_2 = \atree_1$.

    Then (I) is preserved under $(\dstt{1}, \dstt{2})$ if and only if there exists $d \in \{\true, \false\}$ such that $\actexec{}(\dst{b}, \dstt{b}, \batch{b}, a_i) = d$ for $b \in \{1,2\}$ where $a_i.\rol{} = \rol{b}$.
\end{lemma}

\begin{proof}[Proof of \cref{lem:dag-gsc-all-or-none-exec}]
    The ``if'' direction corresponds to the honest execution; its proof is simple and thus omitted. We prove the contrapositive of the ``only if'' direction. Assume there exists $1 \leq i \leq n-1$ and such that (without loss of generality)  $a_{i}'.\rol{} = \rol{1}$ and $a_{i + 1}'.\rol{} = \rol{2}$, and $\actexec{}(\dst{1}, \dstt{1}, \batch{1}, a_{i}') = \true  $ while $\actexec{}(\dst{2}, \dstt{2}, \batch{2}, a_{i+1}') = \false$. Since $\dagcrt$ is a DAG CRT, we know $k_i^2 = |a_i'.\descnext{} | =|a_{i+1}'.\desc{}| = k_{i+1}^1 > 0$.
    Using \cref{claim:a-exec-iff-desc-inserted-dag}, the former implies $a_{i}'.\descnext{} \subseteq^* \ttree_{1}'$ ; the latter implies $a_{i+1}'.\desc{} \not \subseteq^* \atree_{2}'$, which we break down to one of these three disjoint cases:
    \begin{itemize}
        \item Case 1: $a_{i+1}'.\desc{} \not \subseteq \atree_{2}'$. This gives us directly that $\ttree_{1}' \neq \atree_{2}'$.
        
        \item Case 2: $a_{i+1}'.\desc{} \subseteq \atree_{2}'$ and  $\exists j$ such that $a_{i+1}.\desc{}[j] \in \atree_{2}'$ with $\anonce{1}$ and $a_{i+1}.\desc{}[j+1] \in \atree_{2}'$ with $\anonce{2} \neq \anonce{1}$. Since all $a_i'.\descnext{}[j] \in \ttree_{1}'$ are inserted in $\ttree_{1}'$ with the same $\tnonce{}$, then without loss of generality $\anonce{2} \neq \tnonce{}$ and thus $\ttree_{1}' \neq \atree_{2}'$.
        
        \item Case 3: $a_{i+1}'.\desc{} \subseteq \atree_{2}'$, all $a.\desc{}[j] \in \atree_{2}'$ with the same $\anonce{}$, and there exists another subaction $sa^*.\desc{} \notin a_{i+1}'.\desc{}$ such that $sa^*.\desc{} \in \atree_{2}'$ also with $\anonce{}$. Since $a_{i}'.\descnext{} \subseteq^* \ttree_{1}'$, there cannot be any $\descnext{}$ that matches $sa^*.\desc{}$ with the same $\tnonce{} = \anonce{}$ as the subactions in $a_{i}'.\descnext{}$. Thus $\ttree_{1}' \neq \atree_{2}'$.
    \end{itemize}
    
\end{proof}

We also repurpose \cref{lem:svs-gsc-honest-or-reverse-order} and \cref{lem:chain-gsc-honest-order} and their proofs to handle DAG CRTs, but omit stating them. The only difference is that the $\desc{}$ and $\descnext{}$ fields are arrays (instead of elements) and as such they can be empty (instead of null) and form contiguous subsequences of the trigger/action trees (instead of being a member thereof).

Finally, the proof of \cref{thm:dag-crate-forward-atomicity-two-rollups} is similar to the proof of \cref{thm:chain-crate-forward-atomicity-two-rollups}; the only difference is that we rely on DAG-CRT-repurposed versions of \cref{lem:svs-gsc-honest-or-reverse-order} and \cref{lem:chain-gsc-honest-order}.

\newpage
\section{Application: Cross-Rollup Flash Loan}~\label{apdx:xfl}

\begin{lstlisting}[language=Solidity][h]
contract XToken is IERC20 {

    uint256 public totalSupply;
    mapping(address => uint256) public balanceOf;
    GeneralSystemContract public gsc;
    address public oTokenAddr;

    function setup(address _oTokenAddr) public {
        oTokenAddr = _oTokenAddr;
    }

    modifier onlyGSC() {
        require(msg.sender == address(gsc), "Token: only GSC");
        _;
    }

    function transfer(address recipient, uint256 amount)
        external returns (bool)
    {
        balanceOf[msg.sender] -= amount;
        balanceOf[recipient] += amount;
        return true;
    }

    function mint(address to, uint256 amount) onlyGSC external {
        require(gsc.xSender() == oTokenAddr, "Token: only oToken can trigger mint"); // new addition
        _mint(to, amount);
    }

    function burn(address to, uint256 amount) external {
        address from = msg.sender;
        _burn(from, amount);
        bytes memory cdMint = abi.encodeWithSignature(
            "mint(address,uint256)", 
            to, amount
        );
        gsc.trigger(oTokenAddr, cdMint, 10_000);
    }

    function _mint(address to, uint256 amount) internal {
        balanceOf[to] += amount;
        totalSupply += amount;
    }

    function _burn(address from, uint256 amount) internal {
        balanceOf[from] -= amount;
        totalSupply -= amount;
    }
}
\end{lstlisting}

\newpage

\begin{lstlisting}[language=Solidity]
contract XFlashLoan {

    GeneralSystemContract public gsc;
    bool public isBusy;
    Token public token;
    address public otherFLAddr;

    modifier onlyGSC() {
        require(msg.sender == address(gsc), "FlashLoan: only GSC");
        _;
    }

    function setup(address _otherFLAddr) public {
        otherFLAddr = _otherFLAddr;
    }

    // on rollup 1
    function init(uint256 amount, address target, bytes memory cdTarget) public {
        address borrower = msg.sender;
        require (!isBusy);
        isBusy = true;
        uint256 balBefore = token.balanceOf(address(this));
        token.transfer(borrower, amount);

        (bool status, ) = target.call(cdTarget);
        require(status == true);

        bytes memory cdReact = abi.encodeWithSignature(
            "react(address,address,uint256)", 
            address(this), borrower, balBefore
        );
        gsc.trigger(otherFLAddr, cdReact, 10_000);
    }

    // on rollup 2
    function react(address oFLAddr, address borrower, uint256 balBefore) onlyGSC public {
        require(gsc.xSender() == otherFLAddr);
        require (!isBusy);
        bytes memory cdComplete = abi.encodeWithSignature(
            "complete(address,uint256)", 
            borrower, balBefore
        );
        gsc.trigger(oFLAddr, cdComplete, 10_000);
    }

    // on rollup 1
    function complete(address borrower, uint256 balBefore) onlyGSC public {
        require(gsc.xSender() == otherFLAddr);
        require (isBusy);
        uint256 balAfter = token.balanceOf(address(this));
        require(balAfter >= balBefore);
        isBusy = false;
    }
}
\end{lstlisting}

\newpage
\begin{lstlisting}[language=Solidity]

contract XUserFL {

    GeneralSystemContract public gsc;
    Token public token;
    FlashLoan public fl;
    address public othUserFLAddr;

    constructor(address gscAddr){ 
        gsc = GeneralSystemContract(gscAddr);
    }

    function setup(address tokenAddr, address flAddr, address _othUserFLAddr) public {
        token = Token(tokenAddr);
        fl = FlashLoan(flAddr);
        othUserFLAddr = _othUserFLAddr;
    }

    function simpleXFL(uint256 amount, address arbAddr, bytes memory cdArb) public {
        bytes memory cdStep1 = abi.encodeWithSignature(
            "step1(uint256,address,bytes)", 
            amount, arbAddr, cdArb
        );
        fl.init(amount, address(this), cdStep1);
    }

    // on rollup 1
    function step1(uint256 amount, address arbAddr, bytes memory cdArb) public {
        token.burn(othUserFLAddr, amount); 
        bytes memory cdStep2 = abi.encodeWithSignature(
            "step2(uint256,address,bytes)", 
                amount, arbAddr, cdArb
        );        
        gsc.trigger(othUserFLAddr, cdStep2);
    }

    // on rollup 2
    function step2(uint256 amount, address arbAddr, bytes memory cdArb) public {
        require(gsc.xSender() == othUserFLAddr);
        (bool status, bytes memory data) = arbAddr.call(cdArb);
        require(status == true);
        token.burn(othUserFLAddr, amount);
        bytes memory cdStep3 = abi.encodeWithSignature(
            "step3(uint256)", 
                amount
        );
        gsc.trigger(othUserFLAddr, cdStep3);
    }

    // on rollup 1
    function step3(uint256 amount) public {
        require(gsc.xSender() == othUserFLAddr);
        token.transfer(address(fl), amount);
    }

}

\end{lstlisting}

\end{document}